\documentclass{article}

\usepackage{todonotes}
\usepackage{dsfont}
\usepackage[preprint]{neurips_2019}
\usepackage{hyperref}       
\usepackage{url}            
\usepackage{booktabs}       
\usepackage{amsfonts}       
\usepackage{nicefrac}       
\usepackage{microtype}      
\usepackage{amsmath}
\usepackage{mathabx}
\usepackage{amssymb}
\usepackage{amsthm}
\usepackage{float}
\usepackage{mathtools}
\usepackage{relsize}

\newcommand\bh{{\bf h}}
\newcommand\bI{{\bf I}}

\newcommand\bT{{\bf T}}

\newcommand\bX{{\bf X}}

\newcommand\bY{{\bf Y}}

\usepackage{enumerate}

\newcommand\mbH{{\mathbb H}}
\newcommand\mbK{{\mathbb K}}

\newcommand\mbR{{\mathbb R}}

\newcommand\mcS{{\mathcal S}}

\newcommand\tw{{\tilde w}}

\newcommand\seq{\stackrel{\mathclap{\normalfont\mcS}}{=}}

\newcommand\norm[1]{\|#1 \|}

\newcommand\chb{\widecheck {\beta}}

\newcommand\stb{\boldb^{\star}}
\newcommand\ep{\boldsymbol{\epsilon}}

\newcommand\boldb{\boldsymbol{\beta}}
\newcommand\tboldb{\boldsymbol{\tilde{\beta}}}
\newcommand\hboldb{\boldsymbol{\hat{\beta}}}
\newcommand\oboldb{\boldsymbol{\tilde{\boldb}_o}}

\newcommand\tbet{\tilde{\beta}}

\newcommand\normh{\norm{.}_{\mbH}}
\newcommand\normk{\norm{.}_{\mbK}}

\newcommand\haf{\frac{1}{2}}

\newcommand\bi{\beta_{i}}
\newcommand\hbi{\hat{\beta}_{i}}
\newcommand\tbi{\tilde{\beta}_{i}}
\newcommand\sbi{{\beta}_{i}^{\star}}

\newcommand\la{\lambda_{K}}
\newcommand\laa{\lambda_{H}}

\newcommand\sumiS{\sum_{i \in \mcS}}

\newcommand\eit{e_{i}^{\top}}

\newcommand\given[1][]{\:#1\vert\:}

\newcommand\bbeta{\boldsymbol{\beta}}

\newcommand\hsig{{\hat{\small\Sigma}}_{11}}
\newcommand\hsigg{{\hat{\small\Sigma}}_{21}}

\newcommand\sbb{\tilde{\boldsymbol{s}}}

\newcommand\KK{K^{\nicefrac{1}{2}}}
\newcommand\KKn{K^{-\nicefrac{1}{2}}}

\DeclareMathOperator\E{E}

\newtheorem{theorem}{Theorem}

\newtheorem{lemma}{Lemma}
\newtheorem{definition}{Definition}

\newtheorem{assumption}{Assumption}

\title{Adaptive Function-on-Scalar Regression with a Smoothing Elastic Net}

\author{%
  Ardalan Mirshani \\
  Department of Statistics\\
  The Pennsylvania State University\\
  University Park, PA 16802 \\
  \texttt{azm245@psu.edu} \\
  \And
  Matthew Reimherr\thanks{Corresponding Author} \\
  Department of Statistics\\
  The Pennsylvania State University\\
  University Park, PA 16802 \\
  \texttt{mreimherr@psu.edu}
}

\begin{document}

\maketitle


\begin{abstract}
This paper presents a new methodology, called \textit{AFSSEN}, to simultaneously select significant predictors and produce smooth estimates in a high-dimensional function-on-scalar linear model with a sub-Gaussian errors. 
Outcomes are assumed to lie in a general real separable Hilbert space, $\mbH$, while parameters lie in a subspace known as a {\it Cameron Martin space}, $\mbK$, which are closely related to Reproducing Kernel Hilbert Spaces, so that parameter estimates inherit particular properties, such as smoothness or periodicity, without enforcing such properties on the data. 
We propose a regularization method in the style of an adaptive Elastic Net penalty that involves mixing two types of functional norms, providing a fine tune control of 
both the smoothing and variable selection in the estimated model. 
Asymptotic theory is provided in the form of a functional oracle property, and the paper concludes with a simulation study demonstrating the advantage of using AFSSEN over existing methods in terms of prediction error and variable selection. 
\end{abstract}



\section{Introduction}
In recent years, rapid advances in data gathering technologies and new complex modern studies have presented substantial challenges for extracting information from increasingly large and sophisticated data sets. \textit{Functional data analysis}, FDA, is a branch of statistics for conducting statistical inferences on the complicated objects specially in high dimensional spaces \citep{ramsay2007applied,hsing2015theoretical,Kokoszka2017Introduction}.
In addition, the emergence of inexpensive genotyping technologies has produced a substantial need for tools capable of handling large numbers of scalar predictors \citep{bierut2006novel,scott2007genome,repapi2010genome,hebiri2011smooth,algamal2015regularized,craig2018child}. 
In this paper, we consider the function-on-scalar regression problem when the number of predictors is much larger than the number of subjects/units. We present a new approach, called \textit{AFSSEN}, for \textit{Adaptive Function-on-Scalar Smoothing Elastic Net}, which can separately control the smoothness of the underlying functional parameter estimates as well as select important predictors. 

As in classic statistical analyses, the functional linear model, FLM, is one the principle modeling tools when working with functional data \citep{morris2015functional}.
In these cases, at least one of the outcomes or predictors are functional \citep{reiss2010fast}. When the number of predictors is fixed, then techniques for fitting FLM and their statistical properties are now well understood 
\citep{morris2015functional,Kokoszka2017Introduction}
in low dimensional FLM, 
However, in high dimensional cases where the number of predictors are relatively larger than the number of statistical units, little work has been done which most of them are for Scalar-on-Function settings (\cite{matsui2011variable} ; 
\cite{gertheiss2013variable};
\cite{lian2013shrinkage} ;  \cite{fan2015functional}). 

In Function-on-Scalar regression, which is the problem we consider, \citet{chen2016variable} considered a functional least squares with a \textit{Minimax Concave Penalty}, MCP \citep{zhang2010nearly} with fixed number of predictors. A \textit{pre-whitening} technique was used to exploit the within function dependence of the outcomes.
\citet{barber2017function} presented the Function-on-Scalar LASSO, FSL, which combines a functional least squares with an $\mathcal{L}_1$ penalty introduced in a separable Hilbert Space $\mbH$. Since FSL is a convex optimization problem, it is computationally efficient even with a large number of predictors ($I \gg N$). Additionally, FSL estimates achieve optimal convergence rates, but as with  traditional LASSO, these estimates suffer from an asymptotic bias and do not achieve the \textit{functional oracle property}.
\citet{fan2017high} suggested \textit{Adaptive Function-on-Scalar LASSO}, AFSL, which uses a functional least squares with an adaptive $\mathcal{L}_1$ penalty to reduce the bias problem in FSL. They showed AFSL is computationally as efficient as FSL, but achieves a {\it strong functional oracle property.} However, AFSL provides limited control of the smoothness of the functional parameter estimates, which can affect on the prediction error.
 \citet{parodi2017functional} developed the \textit{Functional Linear Adaptive Mixed Estimation}, FLAME, which simultaneously selects important predictors and estimates the smooth parameters. They assume that while the data lie in a general real separable Hilbert space, $\mbH$, the model parameters lie in a Reproducing Kernel Hilbert Space (RKHS), $\mbK$. The RKHS is a subspace of $\mbH$, which can be identified with a linear operator, $K.$
 They demonstrated that FLAME achieved a {\it weak functional oracle property}, meaning it recovered the correct support with probability tending to one and FLAME estimator is equivalent to the oracle estimator only on certain nice projections. To show that FLAME achieved the strong oracle property required stronger structural assumptions. In their framework, they used a \textit{coordinate descent} algorithm which made it computationally very efficient.

The main obstacle for FLAME is having to  simultaneously control the smoothness and sparsity with a single penalty and tuning parameter. In particular, a tuning parameter value that practically works well for smoothing may not work well for variable selection, and vice versa. To address this issue we propose a method that more carefully controls smoothing and sparsity separately. 
We assume the data live in an arbitrary Hilbert space, $\mbH$, but that some linear constraint of the parameters are enforced to lie in a {\it Cameron-Martin} space, CMS, $\mbK$.  CMS are closely related to Reproducing Kernel Hilbert Spaces, RKHS, when $\mbH=L^2[0,1]$, and the two terms are often used interchangeably \citep{bogachev1998gaussian}.
However, since our $\mbH$ will be more general, we refrain from using the term RKHS to avoid confusion. In our approach, AFSSEN, we use an idea similar to the scalar adaptive elastic net penalty \citep{zou2005regularization,zou2006adaptive}, but for functional data. In particular, AFSSEN exploits a combination of a penalized functional least squares and an adaptive smoothing elastic net penalty containing a $\mathcal{L}_1$ term in $\mbH$ for variable selection and a separate $\mathcal{L}_2$ term in $\mbK$ for controlling the smoothness of the estimated parameters. The AFSSEN parameter estimates inherit the nature properties of the kernel function of $\mbK$, such as smoothness or periodicity. We also show that AFSSEN enjoys better mathematical properties than AFSL or FLAME, even when relaxing the Gaussian error assumption to $C$-subgaussian.  In particular, we show that AFSSEN achieves a strong oracle property in both $\mbH$ and the strong norm $\mbK$. We also provide a very fast coordinate descent algorithm in the {\it R} programming language \citep{R},
whose backend is written in {\it C++} \citep{Rcpp}.

The remainder of the paper is organized as follows. In Section \ref{3:back}
we propose some primary materials and main assumptions for the results presented in next sections. Section \ref{3:theorem} provides the framework and demonstrates the strong oracle property for AFSSEN under some non-strong assumptions. In section \ref{3:implementation} we introduce the implementation and numerical illustration including the coordinate descent algorithm and practical considerations. A simulation study with a discussion on comparing the performance of AFSSEN and FLAME in two different smooth and rough scenarios is given in section \ref{3:ems}. You can see a conclusion in Section \ref{3:conclusion}. All mathematical proofs and derivations can be found in the supplemental.

\section{Background and Methodology} \label{3:back}
Throughout this paper we consider $\mbH$ is a real separable Hilbert space with inner product $\langle . , . \rangle_{\mbH}$ and induced norm $\norm{.}_{\mbH}$. Let $K:\mbH \to \mbH$ 
be a compact, positive definite, self-adjoint linear operator, meaning $\langle Kx , x \rangle_{\mbH} > 0$ when $x \neq 0$, $\langle Kx , y \rangle_{\mbH} = \langle x , Ky \rangle_{\mbH}$ for all $x,y\in\mbH$, and it has a finite trace.  
According to the spectral theorem \citep{dunford1963linear}, we can decompose K as $K(x) = \sum_{i=1}^{\infty} \theta_i \langle v_i , x \rangle_{\mbH} v_i$,
where $\{v_1,v_2,\dots\}$ is an orthonormal basis of $\mbH$ and $\theta_1 \geq \theta_2 \geq \dots \geq 0$ is a positive sequence of real numbers. The eigenvalues, $\{\theta_i\}$ and eigenfunctions, $\{v_i\}$ of $K$ induce a subspace of $\mbH$, denoted $\mbK$, called the {\it Cameron-Martin} \citep{bogachev1998gaussian}
space defined as
\begin{align*}
\mbK=\left\{ h \in \mbH ; \ \ \sum\limits_{i=1}^{\infty} \frac{\langle h , v_i \rangle_{\mbH}^2}{\theta_i} < \infty \right\}. 
\end{align*}
Equivalently, $\mbK$ can be viewed as the image $K^{1/2}(\mbH)$ \citep{bogachev1998gaussian}.
Then $\mbK$ is also a Hilbert space under the inner product 
$\langle x , y \rangle_{\mbK} = \sum\limits_{i=1}^{\infty} \dfrac{\langle x , v_i \rangle_{\mbH} \langle y , v_i \rangle_{\mbH}}{\theta_i}$.
The most commonly encountered Hilbert space in FDA is $L^2[0,1],$ which is the space of real valued square integrable functions over $[0,1]$ with corresponding norm 
$\norm{x}_{\mbH}^2 = \int_{0}^{1}x^2(t)dt$. When $\mbH=L^2[0,1]$ and $K$ is an integral operator with kernel $K(t,s)$, then $\mbK$ is isomorphic to a Reproducing Kernel Hilbert Space \citep{berlinet2011reproducing}.

In previous works \citep{fan2017high,parodi2017functional}, the modelling noise was assumed to be a Gaussian process. In this paper, we relax this assumption by consider a $C$-subgaussian noise, 
which, to the best of our knowledge, has not been considered before in functional data models.  
 A mean zero random element $X$ in $\mbH$ is called $C$-subgaussian if
\begin{align*}
\E \left[ \exp{ \langle x , X \rangle_{\mbH} } \right] \leq \exp\left(\frac{1}{2}\langle x , C(x) \rangle_{\mbH} \right) \ \ \ \ \ \forall x \in \mbH,
\end{align*}
where $C$ is a covariance operator in $\mbH$ \citep{buldygin1980sub,Antonioni1997sub}.
In other words, the moment generating function of the $X$ is dominated, uniformly across $\mbH$, by the moment generating function of a Gaussian process with covariance $C$.  This is a convenient assumption as it provides the necessary tail probability inequalities \citep{hsu2012tail} without making explicit assumptions on the distribution of $X$. In particular, one can see a Gaussian process in $\mbH$ with covariance operator $C$ will be a $C$-subgaussian process in $\mbH$ \citep{Antonioni1997sub}.


We now introduce our primary modelling assumptions.
\begin{assumption} \label{assumption:AFSSEN}
Let $Y_{n} \in \mbH \ $ for $n \in 1 , \dots , N$ satisfy
\begin{align*}
Y_{n}=\sum_{i=1}^{I} X_{n,i}\sbi+\epsilon_{n},
\end{align*}
where $\bX=\{X_{n,i}\} \in \mbR^{N \times I}$ is the design matrix with standardized columns and $\epsilon_{n}$ are i.i.d $C$-subgaussian random elements of $\mbH$. We assume that only the first $I_{0}$ predictors are significant, meaning their corresponding coefficient functions, $\beta^{\star}_{1} , \dots , \beta^{\star}_{I_{0}}$ are nonzero. We denote $\bX=(\bX_{1},\bX_{2})$ to partition the predictors into $\bX_{1}$ and $\bX_{2}$ which are called the \textit{significant} and \textit{null} predictors respectively.  The true support we denote as $\mcS = \{1,\dots,I_0\}$.
\end{assumption}
The oracle estimate is defined as $\oboldb = \{\oboldb_{1},\textbf{0}\}$ where $\oboldb_{1}$ is the $\mathcal{L}_2$ penalized estimate given the true support and $\textbf{0} \in \mbH^{I-I_0}$ consists of $I-I_0$ zero functions.
In this paper, we provide an estimator of $\bbeta$ that achieves two types of {\it strong oracle properties,} namely, our estimator asymptotically has the correct support and also is equivalent to the oracle estimator in the $\mbH$ topology as well as the stronger $\mbK$ topology. 

We propose estimating $\bbeta^*$ by minimizing the the following target function over $\mbH^I$
\begin{align} \label{AFSSEN}
L_{\lambda}(\bbeta)=\dfrac{1}{2N} \|\bY-\bX\boldb\|_{\mbH^I}^2+\frac{\la}{2} \sum_{i=1}^{I} \|L(\beta_{i})\|_{\mbK}^2 + \laa \sum_{i=1}^{I} \tw_{i} \| \beta_{i}\|_{\mbH},
\end{align}
where $\bY \in \mbH^{N}$, $\bX \in \mbR^{N \times I}$ and $\boldb \in \mbK^{I}$. 
The operator $L:\mbK \to \mbK$ is a 
continuous linear operator
which is included to provide slightly more generality.
In particular, if one wishes to only penalize the second derivative of $\beta_i$, then that is equivalent to using a Sobolev kernel for $\mbK$ and choosing $L$ as a projection onto the orthogonal compliment of the constant and linear functions (since they have second derivative zero) \citep{bawa2005spline,yuan2010reproducing}.

The estimator produced by minimizing \eqref{AFSSEN} we call \textit{Adaptive Function-on-Scalar Smoothing Elastic Net}, AFSSEN, as it is an extension of the classic Elastic Net \citep{zou2005regularization,zou2006adaptive} to functional response models.
 Setting $\la = 0$, AFSSEN reduces to adaptive function-on-scalar lasso, AFSL, \citep{fan2017high}. Comparing AFSSEN with another approach, FLAME \citep{parodi2017functional} utilizes an $\mathcal{L}_1$ penalty with the $\mbK$ norm, which simultaneously selects significant predictors and produces smooth estimates of the parameters, while in AFSSEN, this task is split between two penalties.  The first is an $\mathcal{L}_1$ penalty using the ${\mbH}$ norm, which is responsible for variable selection, while smoothing is achieved by using an $\mathcal{L}_2$ penalty with the squared ${\mbK}$ norm. 
 We show that tuning sparsity and smoothness of the estimates separately produces stronger asymptotic results and can dramatically increase statistical utility.

The AFSSEN target function \eqref{AFSSEN} requires a kernel, weights $\tw_i$, and the values of penalty parameters $\la$ and $\laa$.
When $\mbH = L^2[0,1]$, there are many options for choosing the kernel functions, each of which imparts different properties to the parameter estimates. We explore four popular kernels:
\begin{align} \label{Kernels}
K_{1}(t,s) & =\exp\left\{\dfrac{-|t-s|^2}{\rho} \right\},  \\
K_{2}(t,s) & =\left(1+\dfrac{\sqrt{5}|t-s|}{\rho}+\dfrac{5(t-s)^2}{3\rho^2}\right)\exp\left\{\dfrac{-\sqrt{5}|t-s|}{\rho}\right\}, \nonumber \\
K_{3}(t,s) & =\left(1+\dfrac{\sqrt{3}|t-s|}{\rho}\right)\exp\left\{\dfrac{-\sqrt{3}|t-s|}{\rho}\right\}, \nonumber \\
K_{4}(t,s) & =\exp\left\{\dfrac{-|t-s|}{\rho}\right\}. \nonumber
\end{align}
Each kernel is from the Mat\'ern family of covariances \citep{stein2012interpolation} with smoothness parameters $\nu=\infty, 5/2, 3/2,1/2$ respectively, though the first is also known as the Gaussian or squared exponential kernel and the last is also known as the exponential, Laplacian, or Ornstein-Uhlenbeck kernel. There are also different options to choose the adaptive weights. One can use the data driven weights $\tw_i=\nicefrac{1}{\norm{\tbet_i}_{\mbH}}$ where  $\tboldb=(\tbet_1,\dots,\tbet_I)^\top$ 
is the FSL parameter estimate \citep{barber2017function}, which has the added benefit of screening out small effects before fitting the more complicated AFSSEN model.
Another option is to run a nonadaptive version of AFSSEN, setting $\tw_i=1$,  then compute the parameter estimate $\hboldb$ and finally define the new weights $\tw_i=\nicefrac{1}{\norm{\hat \beta_i}_{\mbH}}$ for running the adaptive step. 
\citet{huang2008adaptive} suggest using one over the norm of the parameter estimation come from running a marginal regression. The second method, running the nonadaptive step with all wights set to one, is the approach we take in the simulation section. Finally for determining the penalty parameters $\la$ and $\laa$, we consider a fine range and then find their optimal values based on cross-validation, though there are other options, such as BIC \citep{barber2017function}.
\section{Theoritical Properties} \label{3:theorem}
In this section, we present our main theoretical results. We begin by explicitly introducing more technical assumptions needed on the tuning parameters.  We decompose our assumptions into three sets.  The first, Assumption \ref{a:main}, ensures that AFSSEN, asymptotically, recovers the true support of $\bbeta$.  The second, Assumption \ref{assumption:normH}, ensures that AFSSEN is also asymptotically equivalent to the oracle estimate in the $\mbH$ topology, which completes the strong oracle property.  Finally, under Assumption \ref{assumption:normK}, one can show that AFSSEN also achieves the strong oracle property in the stronger $\mbK$ topology, which we have not seen from any other estimator, but is useful for estimating quantities such as derivatives.  
\begin{assumption}  \label{a:main}
Suppose Assumption \ref{assumption:AFSSEN} is satisfied.  Denote the true support as $\mcS$.  We assume the following six conditions hold.
\begin{enumerate} \label{assumption:4p}
\item \textbf{Minimum Signal.} Let $b_{N}=\min\limits_{i \in \mcS} \norm{\beta^{\star}_i}_{\mbH}$, then we assume 
\begin{align}  \label{a:1.1}
b_{N}^2 \gg \frac{I_{0}^2 \log(I)}{N}.
\end{align}
\item \textbf{Sparsity Tuning Parameter.} We assume the sparsity tuning parameter, $\laa$, satisfies
\begin{align}
\frac{\sqrt{I_{0}}\log(I)}{N} \ll \laa \ll \frac{b_{N}^2}{\sqrt{I_{0}}}.
\end{align}
\item \textbf{Design Matrix.} Let $\hsig=N^{-1} \bX_{1}^\top \bX_{1} $, the design matrix for true predictors, then we assume minimum and maximum eigenvalue of $\hsig$ will be bounded by
\begin{align} \label{a:1.3}
\frac{1}{\tau} \leq \sigma_{min}(\hsig) \leq \sigma_{max}(\hsig) \leq \tau,
\end{align}
where $\tau$ is a fixed positive number.
\item \textbf{ّIrrepresentable Condition.} Let $\hsigg=N^{-1} \bX_{2}^\top \bX_{1} $, the cross covariance between the true and null predictors, then we assume that
\begin{align} \label{a:1.4}
\| \hsigg \hsig^{-1} \|_{op} \leq \phi < 1,
\end{align}
where $\norm{A}_{op} = \sup_{\norm{x}=1}\norm{Ax}$ is an operator norm defined for the arbitrary matrix A.
\item \textbf{Maximum Signal.} Let $d_N=\max\limits_{i \in \mcS} \norm{\beta^{\star}_i}_{\mbK}$, we assume the smoothing tuning parameter, $\la$, satisfies
\begin{align}
\la \ll \frac{b_N^2}{I_0 d_N^2}.
\end{align}
\item \textbf{Smoothing Tuning Parameter.} then we assume
\begin{align}
\frac{I_0 \sqrt{\log(I)} d_N}{\sqrt{N}} \ll \frac{\laa}{\sqrt{\la}}.
\end{align}
\end{enumerate}
\end{assumption}
The above assumptions are common in the high dimensional regression literature. The first condition indicates the minimum magnitude of the signals for detecting the relevant predictors. It allows the smallest value of $\norm{\stb_i}_{\mbH}$ vary with the sample size, the number of significant and whole predictors, $I_0$ and $I$, but cannot be too small. The second condition, on the sparsity tuning parameter, states a familiar rate for $\laa$ allowing it to grow but not too fast. The third condition, on the design matrix, guarantees that the oracle estimator is well defined, which, in turn ensures that the AFSSEN estimates are well behaved when restricted to the true predictors. The Irrepresentable condition implies that the true and null predictors should not be too correlated. This is an essential assumption for achieving the oracle property \citep{zhao2006model}. The fifth condition, on the maximum signal, essentially indicates that the smoothing tuning parameter, $\la$, cannot be increased too quickly. Finally, the last condition gives a trade-off between the smoothing and sparsity parameters. It indicates that the sparsity parameter cannot be too small relative to the smoothing parameter.

The above assumptions will imply that AFSSEN is consistent in terms of variable selection. We require slightly stronger assumptions to show the AFSSEN estimates are asymptotically equivalent to the oracle estimates under $\norm{.}_{\mbH}$ and $\norm{.}_{\mbK}$.
\begin{assumption} \label{assumption:normH}
The sparsity tuning parameter $\laa$, satisfies
\begin{align*}
\laa \ll \frac{b_N}{\sqrt{N} \sqrt{I_0}}.
\end{align*}
\end{assumption}
\begin{assumption} \label{assumption:normK}
Assume that $\eta_j^2 \geq M \sqrt{\theta_j}$ where $\theta_j$ and $\eta_j$ are the eigenvalues of the $K$ and $L$, respectively, and $M>0$ is a constant scalar, then the smoothing and sparsity tuning parameters $\la,\laa$ satisfy
\begin{align*}
\frac{\laa}{\la} \ll \frac{b_N}{\sqrt{N} \sqrt{I_0}}.
\end{align*}
\end{assumption}
Assumption \ref{assumption:normH} assigns a tighter upper bound than the Sparsity Tuning Parameter condition in Assumption \ref{assumption:4p}. Assumption \ref{assumption:normK} gives another trade-off between $\la$ and $\laa$ and does not allow their ratio grow really fast.
Now with using the above assumptions, we can present our main theorem which shows the AFSSEN chooses the true support with probability one and their estimates are asymptotically equivalent with the oracle estimates.
\begin{theorem} \label{t:w.oracle}
Suppose $\tboldb$ and $\oboldb$ are the FSL \citep{barber2017function} and oracle estimates respectively. Assume $L$ is a self-adjoint nonnegative definite continuous linear operator with the same eigenfunctions as $K$.  Let $\hboldb$ be the the AFSSEN estimate with the data driven weights $\tw_i=\|\tbet_i\|_{\mbH}^{-1}$. If the regression model satisfies Assumptions \ref{assumption:AFSSEN} and \ref{assumption:4p}, the AFSSEN estimates $\hboldb$
\begin{enumerate}
\item has the correct support:
\begin{align*}
P(\hboldb \seq \stb) \rightarrow 1,
\end{align*}
\item is equivalent to oracle estimate under $\norm{.}_{\mbH}$ if Assumption \ref{assumption:normH} also holds:
\begin{align*}
\| \hboldb - \oboldb\|_\mbH =o_{P}(N^{-1/2}),
\end{align*}
\item and is equivalent to oracle estimate under $\norm{.}_{\mbK}$ if Assumption \ref{assumption:normK} also holds:
\begin{align*}
\|\hboldb - \oboldb\|_{\mbK}=o_{P}(N^{-1/2}).
\end{align*}
\end{enumerate}
\end{theorem}
\section{Implementation} \label{3:implementation}
Here we present a coordinate descent algorithm to find the parameter estimates efficiently. We employ functional subgradients to update the individual parameter estimates in each step. Subgradients extend derivatives to non necessarily differentiable convex functionals.
We call $h \in \mbH$ a subgradient of $f$ at $x_0 \in \mbH$ if
\begin{align} \label{intro:derivative:subg}
f(x) \geq f(x_0) + \langle h , x - x_0 \rangle_{\mbH} \quad \quad \forall x \in \mbH.
\end{align}
The collection of the all subgradients of $f$ at $x_0 \in \mbH$ is called the subdifferential of $f$ at $x_0$ and denoted by $\partial f(x_0)$.
It is clear from \eqref{intro:derivative:subg} that if $0 \in \partial f(x_0)$, then $x_0$ is a minimizer of $f$. 
For more details and background we refer interested readers to  \citet{boyd2004convex,bauschke2011convex,barbu2012convexity,shor2012minimization}. We show in the supplemental material that the subgradient of the target function \eqref{AFSSEN} is
\begin{align*}
\frac{\partial L_{\lambda}(\beta)}{\partial \beta_{i}}=K(- N^{-1} \bX_{.i}^\top (\bY-\bX\boldb))+\la L^2(\bi) + \laa \tw_{i} \left\{
                \begin{array}{ll}
                K(\beta_{i}) \norm{\beta_{i}}_{\mbH}^{-1} & \bi \neq 0  \\
                  \\ \\
\{h ; \  \small{ \norm{\KKn(h)}_{\mbK} \leq 1} \}  & \bi =  0, \\
                \end{array}
              \right.
\end{align*}

where $\bX_{.i}^\top=\left( X_{1i},\dots,X_{Ni} \right) \in \mbR^{N}$ is the $i^{th}$ column of the design matrix $\bX$. Then we can conclude the following useful lemma.
\begin{lemma} \label{l:hbeta}
The AFSSEN estimate satisfies the equations
\begin{align*}
\left\{
\begin{array}{ll}
    \hbi=0 \qquad & \quad  \norm{\chb_{i}}_{\mbH} \leq \laa \tw_{i}\\
    \hbi=\left((1+\dfrac{\laa \tw_{i}}{\|\hat{\beta}_i\|_{\mbH}})\mathds{I} 
    + \la K^{-1} L^2 \right)^{-1} \chb_{i} \qquad & \quad \| \chb_{i}\|_{\mbH} > \laa \tw_{i} \\
\end{array}
\right.
\end{align*}
where $\mathds{I}$ is the identity operator from $\mbH$ to $\mbH$ and $\chb_{i} = N^{-1} {\sum\limits}_{n=1}^{N} \ X_{ni} (Y_{n} - {\sum\limits}_{j \neq i} \ X_{nj} \hat \beta_{j}))$.
\end{lemma}
The only challenge in using the Lemma \ref{l:hbeta} is the presence of $ \norm{\hat{\beta}_i}_{\mbH}$
which the following Lemma \ref{l:num} can help us to derive an estimation for it.

\begin{lemma} \label{l:num}
The $\norm{\hat{\beta}_i}_{\mbH}$ in Lemma \ref{l:hbeta} can be solved numerically by
\begin{align*}
1 = \sum_{j=1}^{\infty} \frac{ \langle \chb_{i} , v_{j} \rangle^2 }{\left( (1 +\la\eta_j^2 \theta_j^{-1} )\norm{\hat{\beta}_i}_{\mbH} + \laa \tw_{i}\right)^2},
\end{align*}
\end{lemma}
where the $\eta_j$ and $\theta_j$ are the eigenvalue of the $L$ and $K$ operators respectively for common eigenfunction $v_i$.
Now one can run the coordinate descend algorithm iteratively and obtain a sequence $\hboldb^{(t)}$ from the estimated parameters which converges to desired $\hboldb$ asymptotically. 

In practice, we follow an approach similar to the one outlined in FLAME \citep{parodi2017functional}.  We run our algorithm in \textit{nonadaptive} and \textit{adaptive} steps. First, in the  nonadaptive step, we set all $\tw_i=1$ and find the estimated $\hat{\beta}_j^{ndp}$. Then for adaptive step, we set $\tw_i=\nicefrac{1}{\norm{\hat{\beta}_j^{ndp}}_{\mbK}}$. 
To choose penalty parameters,
we select $\la$ from $\{10,1,0.01,0.0001,0\}$ and $\laa$ from 100 points between $\lambda_{max}$ to $r_{\lambda}\lambda_{max}$ where $\lambda_{max}$ is the smallest tuning parameter such that all parameters are set to zero, while $r_{\lambda}$ is a specified ratio. 
In order to increase the computational efficiency, for any fixed $\la$, we start with $\laa=\lambda_{max}$ and initial $\boldsymbol{\beta}=0$. Then we decrease the $\laa$ and in each step, using a \textit{warm start} which means the previous estimated $\boldsymbol{\hat{\beta}}$ is used as the initial value of $\boldsymbol{\beta}$.
We employ a \textit{kill switch} variable, where this iterative process is stopped once if the number of active predictors exceeds a chosen threshold (since one is search for spase solutions). Small changes in $\laa$ combined with a warm start imply a very quick convergence of $\boldsymbol{\hat{\beta}}$ in each step.
It is also more efficient to define a maximum number of iterations $T$ for $t$ and a \textit{threshold} as stopping criteria on the improvement parameter estimation 
$\norm{\boldsymbol{\hat{\beta}}^{(t)} - \boldsymbol{\hat{\beta}}^{(t-1)}}_{\mbH^{I}}$.
We use a 10-fold cross validation to find the optimum values of $\la$ and $\laa$. Finally, we run 100 iterations for each setting to find the average of prediction error
$\left(\sum_{n=1}^{N} \norm{\bX_{n.}^{\top} \boldsymbol{\hat{\beta}} - \bX_{n.}^{\top} \boldsymbol{{\beta}}^{\star} }_{\mbH} \right)$, prediction error derivatives 
$\left(\sum_{n=1}^{N} \norm{\bX_{n.}^{\top} \boldsymbol{\hat{\beta}^{'}}- \bX_{n.}^{\top} {\boldsymbol{\beta}^{\star}}^{'}}_{\mbH}\right)$ and the number of true and false positive predictors. 
\section{ٍEmpirical Study} \label{3:ems}
In this section, we compare the performance of AFSSEN with FLAME in two high-dimensional simulation settings, one with rougher and and one with smoother $\beta^{\star}$ coefficients. 
Mimicking FLAME, we generate $N=500$ functional observations from $\mbH = L^2[0,1]$ and $I = 1000$ scalar predictors, with $I_0 = 10$ significant. The design matrix ${\mathbf X}$ is generated 
using standard normal random variables. Observation errors $\varepsilon_n(t)$ are generated according to a 0-mean Matern process with parameters $(\nu = \nicefrac{3}{2}, \textrm{range} = \nicefrac{1}{4}, \sigma^2=1)$.

Here we consider the four different RKHS kernels, K, in \eqref{Kernels} with varying range parameters: $\{0.5,1,2,4,8,16,32\}$. Denote $\theta_1 \geq \theta_2 \geq \dots \geq 0$ and $v_1, v_2, \dots \in \mbH$ as the ordered eigenvalues and their corresponding eigenfunctions of $K$ and use the eigenfunctions, computed numerically on a grid of $m=50$ evenly spaced points between 0 and 1, as an orthonormal basis of $\mbH$. They allow us to compute the $\norm{.}_{\mbH}$ and $\norm{.}_{\mbK}$ quickly. In order to have more computational efficiency, we take the number of FPCs that explain more than $99\%$ of the variability in FPCA, which means $\sum_{i=1}^{M} \theta_i \geq 0.99 \sum_{i=1}^{\infty} \theta_i$.
We also considered $r_{\lambda}=10^{-6}$ and a $0.001$ threshold as the stopping criteria for the coefficient increments ($\norm{ \boldsymbol{\hat{\beta}}^{(T)} - \boldsymbol{\hat{\beta}}^{(T-1)}}_{\mbH^I} \leq 0.001$) and a kill switch $2I_0=20$ for maximum number of non zero predictors before we stop decreasing $\laa$. 

\subsection{Rough Setting}
In this scenario, The true coefficients $\beta_i^{\star}(t)$ are sampled from a Matern process with 0 average and parameters $(\nu = \nicefrac{5}{2}, \textrm{range}  = \nicefrac{1}{4}, \sigma^2=1)$. The results are presented in Figure \ref{Rough_AFSSEN}.
Turning to average of prediction error and prediction error derivative, the AFSSEN performs 5 to 10 times better than FLAME. They also looks to be consistent in terms of range parameter for all kernels which is a significant improvement than FLAME. The behavior of AFSSEN in variable selection is not as much consistent and seems to work better for smaller range parameters and can beat FLAME in those cases. It seems the FLAME is the winner for larger range parameters. However, in all AFSSEN situations, the average of false positive number of predictors are still remain less than one which shows a quite small uncertainty. Lastly, the rougher kernel, i.e. exponential, in AFSSEN seems to be more efficient than the others for rough $\beta_i^{\star}$ predictors.
\begin{figure}
	\centering
	\includegraphics[width=10.8cm,height=6.7cm]{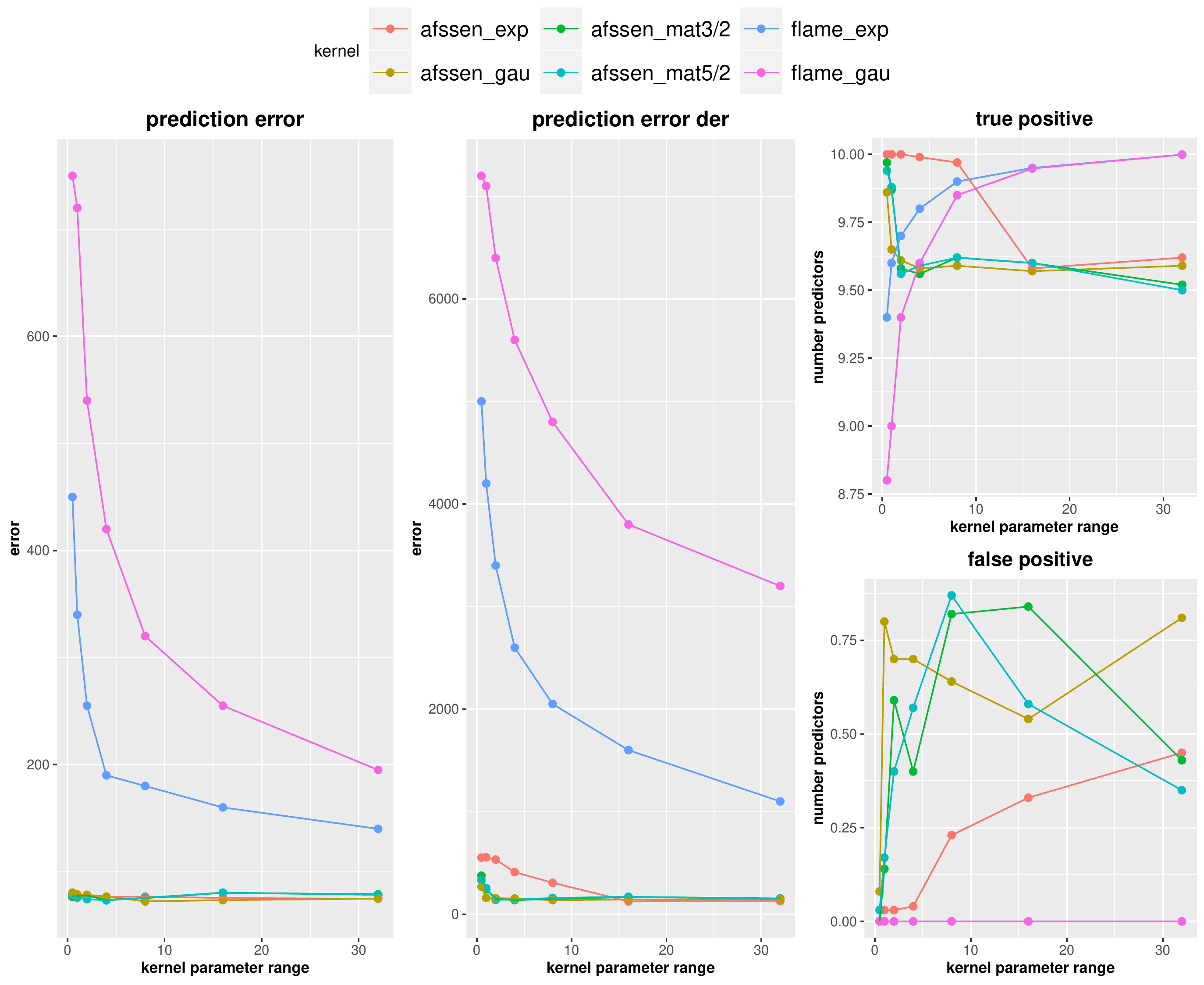}
	\caption{Summary of the simulations varying kernel for the rough case.}
	\label{Rough_AFSSEN}
\end{figure}

\subsection{Smooth Setting}
For the smooth setting, we just generate the true coefficients from a Matern process with 0 average and parameters $(\nu=\nicefrac{7}{2},$range$=1,\sigma^2=1)$ and keep the other parameters same as rough setting. Figure \ref{Smooth_AFSSEN} illustrates the AFSSEN performs $50-200\%$ better than FLAME in prediction error and prediction derivative error. For false positive predictors, the AFSSEN beats FLAME but still cannot force them to be zero. In the number of true positive predictors, the AFSSEN works better for smaller range parameters but not for the larger ones. In the smooth setting, it seems using the smoother kernels implies the smaller prediction errors but performs worth in variable selections. A final remark is consistency of AFSSEN than FLAME in prediction errors and also number of true positive predictors. However in terms of false positive, FLAME is more consistent but has higher values than AFSSEN.
\begin{figure}
	\centering
	\includegraphics[
 	width=10.8cm,height=6.7cm]{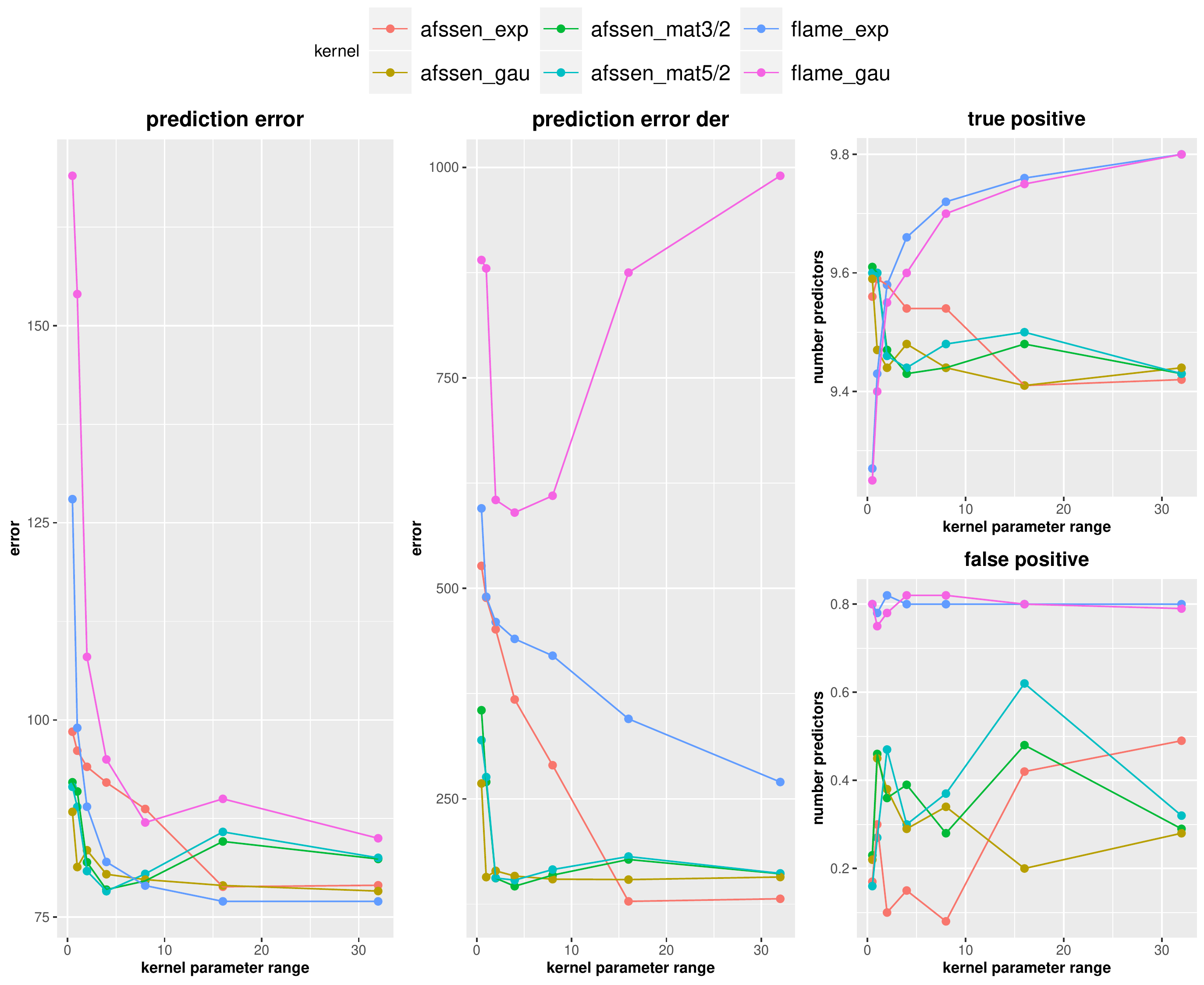}
	\caption{Summary of the simulations varying kernel for the smooth case. }
\label{Smooth_AFSSEN}
\end{figure}  
\section{Conclusion} \label{3:conclusion}
We have presented a method, called AFSSEN, which can control the dimension reduction and smoothness of the parameter estimates in a high-dimensional function-on-scalar linear model with a sub-Gaussian errors. In our work, the parameters live in a Reproducing Kernel Hilbert Space (RKHS), $\mbK$,
and inherit its properties, such as smoothness or periodicity. In our framework, the data is not enforced to lie in the RKHS. We showed under some non-strong assumptions, our parameter estimates and the true parameters have the same support and then illustrated the strong functional oracle property would be achieved under both norm $\mbH$ and norm $\mbK$. 
Using a simulation study, we depicted a hugely improvement on prediction error and prediction error derivative and consistency using AFSSEN than previous works. Additionally, in terms of true and false positive error, we showed AFSSEN beats FLAME in smooth coefficient parameters and have a highly reliable performance in the rough cases.
\bibliographystyle{chicago}
\addcontentsline{toc}{chapter}{Bibliography}
\bibliography{AFSSEN}

\newpage

\section*{Supplementary Material}
In this section, we provide the proof of lemmas and theorems discussed in the main content.
We start defining some notations which are necessary throughout this section.
\begin{definition} \label{d:notation}
Let $K: \mbH \to \mbH$ be an operator.  We define the coordinate wise extension, $K_M$ from $\mbH^M$ to $\mbH^M$ as
\begin{align*}
K_{M}(\bh) := 
\left(K(h_1),\dots,K(h_M)\right) \in \mbH^{M} \qquad \text{where } \bh=(h_1,\dots,h_M) \in \mbH^{M}.
\end{align*}
We define $L_M: \mbK^M \to \mbH^M$ analogously
\end{definition}
\begin{definition} \label{d:notationT}
Let $\Sigma \in \mbR^{M \times M}$ be a matrix.  Then, using an abuse of notation, we define the linear operation $\Sigma: \mbH^M \to \mbH^M$ as
\[
\Sigma \bh := \left\{\sum_{j=1}^M \Sigma_{1j}h_j, \dots,\sum_{j=1}^M \Sigma_{Mj}h_j \right\}.
\]
\end{definition}
Note that, as defined, the operators $\Sigma$ and $K_M$ are interganchable in the sense that $\Sigma K_M \bh = K_M \Sigma \bh$. The following lemma are needed for demonstrating the main theoretical properties.
\begin{lemma} \label{l:subgrad}
Let $\partial f(x)$ denote the subdifferential of a functional $f:\mbK \to \mbR$ at $x$.  Then we have the following.
\begin{enumerate}
\item Consider the functional $f(x)=\norm{x}_{\mbH}^2$. Then $f$ is convex and everywhere differrentiable with respect to $\norm{.}_{\mbK}$ with 
\begin{align*}
\partial f(x)=2K(x).
\end{align*}
\item Consider the functional $f(x)=\norm{x}_{\mbH}$. Then $f$ is convex and differrentiable with respect to $\norm{.}_{\mbK}$ when $x \neq 0$ with
\begin{align*}
\partial f(x)=K(x) \norm{x}_{\mbH}^{-1},
\end{align*}
and when $x=0$ with
\begin{align*}
\partial f(0)= \{ h \in \mbH ; \  \| K^{-\nicefrac{1}{2}} (h) \|_{\mbK} \leq 1 \}.
\end{align*}
\item Consider the functional $f(x)=\norm{L(x)}_{\mbK}^2$ where $L$ is a self-adjoint linear operator from $\mbK$ to $\mbK$. Then $f$ is convex and everywhere differrentiable with respect to $\norm{.}_{\mbK}$ with
\begin{align*}
\partial f(x)= 2L^2(x).
\end{align*}
\end{enumerate}
\end{lemma}
\begin{proof} 
Recall if $f: \mbK \rightarrow \mbR$ is a convex functional, $h \in \mbK$ is called a subgradient of $f$ in $x \in \mbK$ with respect to $\norm{.}_{\mbK}$ when
\begin{align*}
f(y)-f(x) \geq \langle h , y - x \rangle_{\mbK} \quad \quad \forall y \in \mbK,
\end{align*} and
the collection of the all subgradients of $f$ at $x \in \mbK$ is called the subdifferential of $f$ in $x$ and denoted by $\partial f(x)$. \\ \\
Part 1: According to the fact that $\norm{x}_{\mbH}^2=\norm{\KK(x)}_{\mbK}^2$, we need to prove
\begin{align*}
 \| \KK(y) \|_{\mbK}^2 - \| \KK(x) \|_{\mbK}^2 \geq \langle 2K(x) , y-x \rangle_{\mbK}. \\
\end{align*}
The right hand side can be written as
\begin{align*}
\langle 2K(x) , y-x \rangle_{\mbK} = 2 \langle \KK(x) , \KK(y) \rangle_{\mbK} - 2\| \KK(x) \|_{\mbK}^2,
\end{align*}
and the left hand side is
\begin{align*}
 \| \KK(y) \|_{\mbK}^2 + \| \KK(x) \|_{\mbK}^2 -2 \| \KK(x) \|_{\mbK}^2.
\end{align*}
So \textit{Cauchy Schwarz} inequality gives the desired result. \\ \\
Part 2: For $x \neq 0$, we need to show that
\begin{align*}
 \| \KK(y) \|_{\mbK} - \| \KK(x) \|_{\mbK} \geq \langle \frac{K(x)}{ \| \KK(x) \|_{\mbK}} , y-x \rangle_{\mbK}, \\
\end{align*}
or equivalently
\begin{align*}
 \| \KK(y) \|_{\mbK}\| \KK(x) \|_{\mbK} - \| \KK(x) \|_{\mbK}^2 \geq \langle K(x) , y-x \rangle_{\mbK} = \langle \KK(x) , \KK(y) \rangle_{\mbK} -\norm{\KK(x)}_{\mbK}^2,\\
\end{align*}
which is true based on \textit{Cauchy Schwarz} inequality. 

Let's assume $x=0$. We should find all $h \in \mbH$ such that
\begin{align*}
 \| \KK(y) \|_{\mbK} - 0 \geq \langle h , y - 0 \rangle_{\mbK}. \\
\end{align*}
So based on the following application of \textit{Cauchy Schwarz} inequality
\begin{align*}
\langle h , y \rangle_{\mbK}=\langle \KKn(h) , \KK(y) \rangle_{\mbK} \leq \|\KKn(h)\|_{\mbK} \|\KK(y)\|_{\mbK}. \\
\end{align*}
The part 2 trivially holds when $\norm{\KKn(h)}_{\mbK} \leq 1 $. \\ \\
Part 3. It is enough to show that
\begin{align*}
 \norm{L(y)}_{\mbK}^2 - \norm{L(x)}_{\mbK}^2 \geq \langle 2L^2(x) , y-x \rangle_{\mbK}. \\
\end{align*}
The right hand side can be written as
\begin{align*}
\langle 2L^2(x) , y-x \rangle_{\mbK} = 2 \langle L(x) , L(y) \rangle_{\mbK} - 2\norm{L(x)}_{\mbK}^2.
\end{align*}
Same as part 1, the inequality is satisfied with using the \textit{Cauchy Schwarz} inequality.

\end{proof}
\begin{lemma} \label{l:subdiff}
The subgradient of target function \eqref{AFSSEN} is
\begin{align*}
\frac{\partial L_{\lambda}(\beta)}{\partial \beta_{i}}=K(-N^{-1} \bX_{.i}^\top (\bY-\bX\boldb))+\la L^2(\bi) + \laa \tw_{i} \left\{
                \begin{array}{ll}
                K(\beta_{i}) \norm{\beta_{i}}_{\mbH}^{-1} \qquad & \bi \neq 0  \\
                  \\ \\
\{h ; \  \small{ \norm{\KKn(h)}_{\mbK} \leq 1} \} \qquad & \bi =  0 \\
                \end{array}
              \right.
\end{align*}
where $\bX_{.i}^\top=\left( X_{1i},\dots,X_{Ni} \right) \in \mbR^{N}$ is the vector of $i^{th}$ column of design matrix $\bX$.
\end{lemma}
\begin{proof}
\begin{align*}
L_{\lambda}(\beta)& =\dfrac{1}{2N} \|\bY-\bX\boldb\|_{\mbH}^2+\frac{\la}{2} \sum_{i=1}^{I} \|L(\beta_{i})\|_{\mbK}^2 + \laa \sum_{i=1}^{I} \tw_{i} \| \beta_{i}\|_{\mbH} \\
& = \dfrac{1}{2N} \sum\limits_{n=1}^{N} \| Y_{n}-\bX_{n.}^\top \boldb\|_{\mbH}^2+\frac{\la}{2} \sum_{i=1}^{I} \|L(\beta_{i})\|_{\mbK}^2 + \laa \sum_{i=1}^{I} \tw_{i} \norm{\beta_{i}}_{\mbH}, \\
\end{align*}
where $Y_n \in \mbH$ is the $n^{th}$ observation and $\bX_{n.}^\top=\left(X_{n1},\dots,X_{nI} \right) \in \mbR^{I}$ is the $n^{th}$ row of the design matrix $\bX$. According to Lemma \ref{l:subgrad}, we can take the subgradient of $L_{\lambda}(\beta)$ for any $\bi$ with respect to $\norm{.}_\mbK$ as follows
\begin{align*}
& \frac{\partial L_{\lambda}(\beta)}{\partial \beta_{i}} = \dfrac{1}{2N} \sum\limits_{n=1}^{N} \frac{\partial \norm{Y_{n}-\bX_{n.}^\top \boldb }_{\mbH}^2 }{\partial (Y_{n}-\bX_{n.}^\top \boldb)} . \frac{\partial (Y_{n}-\bX_{n.}^\top \boldb)}{\partial \bi}  +\frac{\la}{2} \frac{\partial \| L(\bi)\|_{\mbK}^2}{\partial \bi} + \laa \tw_{i} \frac{\partial \norm{\beta_i}_{\mbH}}{\partial \bi} \\ \\
& = \dfrac{1}{2N} \sum\limits_{n=1}^{N} 2K(Y_{n}-\bX_{n.}^\top\boldb)(-X_{ni})  +\la L^2(\bi) + \laa \tw_{i}  \left\{
                \begin{array}{ll}
                  K(\beta_{i}) \norm{\beta_{i}}_{\mbH}^{-1} \qquad & \bi \neq 0  \\
                  \\ \\
\{h ; \  \small{ \norm{\KKn(h)}_{\mbK} \leq 1} \} \qquad & \bi =  0 \\
                \end{array}
              \right. \\ \\
 & \hspace{20mm} =K (-N^{-1} \bX_{.i}^\top (\bY-\bX\boldb))+\la L^2(\bi) + \laa \tw_{i} \left\{
                \begin{array}{ll}
                 K(\beta_{i}) \norm{\beta_{i}}_{\mbH}^{-1} \qquad & \bi \neq 0  \\
                  \\ \\
\{h ; \  \small{ \norm{\KKn(h)}_{\mbK} \leq 1} \} \qquad & \bi =  0 \\
                \end{array}
              \right.                        
\end{align*}
\end{proof}
Now we introduce the lemma which will play an important role in proof of the functional oracle property.
\begin{lemma} \label{l:hbeta2}
Let's assume the AFSSEN estimation $\hboldb$ and true parameters $\boldb^{\star}$ have the same support, $\mcS = \{1,\dots,I_0\}$.
The nonzero parts of $\hboldb=(\hboldb_1,\textbf{0})$ can be written
concisely by
\begin{align*}
\hat{\boldb}_{1} = \hsig G_{I_0} (\boldb^{\star}_{1}) + G_{I_0} \left(  N^{-1} \bX_{1}^\top \ep - \laa \sbb \right),
\end{align*}
where 
\begin{align*}
& \hsig = \dfrac{1}{N} \bX_1^\top \bX_1 \in \mbR^{I_0 \times I_0}, \\
& \sbb=\{ \tilde{w}_{i} \hbi \norm{\hbi}_{\mbH}^{-1} ; i \in \mcS \} \in \mbK^{I_0}, \\
& G_{I_0} = \left(\hsig \mathds{I}_{I_0}  + \la K_{I_0}^{-1} L^2_{I_0} \right)^{-1}.
\end{align*}
\end{lemma}
\begin{proof}
Let's denote $\boldb^{\star}=(\boldb^{\star}_1,\textbf{0})$ with true support $\mcS=\{1, \dots, I_{0}\}$. We assumed the AFSSEN estimation $\hboldb=(\hboldb_1,\textbf{0})$ have the same support as $\boldb^{\star}$, so we can consider $\hbi \neq 0$ for all $i \in \mcS$  and $\hbi = 0$ for $i \not\in \mcS$. Since $\hboldb$ is going to be the minimizer of the convex function \eqref{AFSSEN}, according to Lemma \ref{l:subdiff}, 
for $i \not\in \mcS$ 
\begin{align*}
 K(\frac{1}{N} \bX_{.i}^\top (\bY-\bX\hboldb)) - \laa \tw_{i}h = 0 \hspace{20mm}  \{ h ; \norm{\KKn(h)}_{\mbK} \leq 1\}, \\
 \KK (\frac{1}{N} \bX_{.i}^\top (\bY-\bX\hboldb))=\laa \tw_{i} \KKn(h) \hspace{20mm}  \{ h ; \norm{\KKn(h)}_{\mbK} \leq 1\}.
\end{align*}
So the above equality exists when
\begin{align*}
\norm{\KK (\frac{1}{N} \bX_{.i}^\top (\bY-\bX\hboldb))}_{\mbK} \leq \laa \tw_{i},
\end{align*}
or equivalently
\begin{align*}
\norm{\frac{1}{N} \bX_{.i}^\top (\bY-\bX\hboldb)}_{\mbH} \leq \laa \tw_{i}.
\end{align*}
In the other side, when $ \norm{\dfrac{1}{N} \bX_{.i}^\top (\bY-\bX\hboldb)}_{\mbH} > \laa \tw_{i}$, we will have $\hbi \neq 0$ for $i \in \mcS$ and then
\begin{align*}
K(-\frac{1}{N} \bX_{.i}^\top (\bY-\bX\hboldb)) +\la L^2(\hbi) + \laa \tw_{i} K(\hbi) \norm{\hbi}_{\mbH}^{-1} =0.
\end{align*}
According to Definition \ref{d:notation}, $\hsig=\frac{1}{N} \bX_1^\top \bX_1 \in \mbR^{I_0 \times I_0}$ and $\sbb=\{ \tilde{w}_{i} \hbi \norm{\hbi}_{\mbH}^{-1} ; i \in \mcS \}$ we have
\begin{align*}
& K_{I_0}(-\dfrac{1}{N} \bX_{1}^{\top} (\bY-\bX_1\hboldb_1)) + \la L_{I_0}^2(\hboldb_1) + \laa K_{I_0}(\sbb) = 0, \\ \nonumber 
& K_{I_0}(-\frac{1}{N} \bX_{1}^{\top} \bY + \hsig \hboldb_1) + \la L_{I_0}^2(\hboldb_1) + \laa K_{I_0}(\sbb) = 0.
\end{align*}
According to Definition  \ref{d:notationT}, we can simplify it by
\begin{align*}
& \left( \hsig K_{I_0} + \la L_{I_0}^2 \right) \hboldb_1 =  K_{I_0}(N^{-1} \bX_{1}^\top \bY) - \laa K_{I_0}(\sbb),
\\
& \left( \hsig \mathds{I}_{I_0} + \la K_{I_0}^{-1} L_{I_0}^2 \right) \hboldb_1 =  N^{-1} \bX_{1}^\top \bY - \laa \sbb,
\end{align*}
then
\begin{align*}
\hboldb_1 = \left( \hsig  \mathds{I}_{I_0} + \la K_{I_0}^{-1} L_{I_0}^2 \right)^{-1} \left(N^{-1} \bX_{1}^\top \bY - \laa \sbb \right).
\end{align*} 
With substitution of $\bY=\bX_{1} \stb_{1} + \ep$, we will have
\begin{align*}
\hboldb_{1} & =\left( \hsig  \mathds{I}_{I_0} + \la K_{I_0}^{-1} L_{I_0}^2\right)^{-1}\left(\hsig \stb_{1} + \left(N^{-1} \bX_{1}^\top \ep - \laa \sbb \right) \right).
\end{align*}
Finally with introducing $G_{I_0}=\left( \hsig  \mathds{I}_{I_0} + \la K_{I_0}^{-1} L_{I_0}^2\right)^{-1}$ as a linear operator from $\mbH^{I_0}$ to $\mbH^{I_0}$, we will have
\begin{align*}
\hboldb_{1} & = \hsig G_{I_0} (\stb_{1}) + G_{I_0} \left(  N^{-1} \bX_{1}^\top \ep - \laa \sbb \right).
\end{align*}

\end{proof}
We now introduce the following lemmas which are useful in proof of the Theorem \ref{t:w.oracle}.
\begin{lemma} \label{l:multi-sub}
Let $\ep=(\epsilon_1,\dots,\epsilon_N)$ where $\epsilon_i$s are independent mean zero $C$-subgaussian process in $\mbH$ and $\bT$ is an arbitrary operator from $\mbH^N$ to $\mbH$, then $\bT \ep$ will be a $C_T$-subgaussian process in $\mbH$ with
\begin{align*}
C_T= \bT C_N {\bT^{*}}^{\top}.
\end{align*}
\begin{proof}
Since $\epsilon_i$s are independent, we can write
\begin{align*}
& \E \exp\left( { \langle x , \bT \ep \rangle_{\mbH} }\right) = \E \exp\left( { \sum\limits_{i =1}^{N} \langle x , T_i \epsilon_i \rangle_{\mbH} } \right) = \E \exp \left({\langle T_1^{*}  x , \epsilon_1 \rangle }\right) \dots \E \exp\left({ \langle T_N^{*} x , \epsilon_N \rangle }\right) \\
& \leq \exp\left({\frac{1}{2} \langle T_1^{*} x , C(T_1^{*} x) \rangle_{\mbH} }\right) \dots \exp\left( {\frac{1}{2} \langle T_N^{*} x , C(T_N^{*} x) \rangle_{\mbH} } \right) \\ 
& = \exp\left({\frac{1}{2} \langle x , T_1C(T_1^{*}x) \rangle_{\mbH}}\right) \dots \exp\left({\frac{1}{2} \langle x , T_NC(T_N^{*}x) \rangle_{\mbH}}\right)
= \exp\left({\frac{1}{2} \langle x , \bT C_N {\bT^{*}}^{\top}(x) \rangle_{\mbH}}\right).
\end{align*}
\end{proof}
\end{lemma}

The following lemma can be considered as an extension of the lemma used in \citep{parodi2017functional} for $C$-subgaussian noise.
\begin{lemma} \label{l:subg-Basic}
Let's consider $X$ is a mean zero $C$-subgaussian process in Hilbert space $\mbH$, then we have
\begin{align*}
P\left(\norm{ X }_{\mbH}^2 \geq \left(\| C\|_{1} + 2 \|C \|_{2} \sqrt{t} + 2 \| C\|_{\infty} t \right) \right) \leq e^{-t},
\end{align*}
where $\|C\|_{1}$ , $\|C \|_{2}$ and $ \|C\|_{\infty}$ represent $\sum\limits_{i=1}^{\infty} \gamma_{i}$, $\sqrt{\sum\limits_{i=1}^{\infty} \gamma_{i}^2}$ and $\max\limits_{i}{\gamma_{i}}$ respectively when $\gamma_{i}$ are eigenvalues of Covariance operator $C$.
\begin{proof}
The idea is same as \citep{barber2017function}. Let's assume $\gamma_i > 0$ and $\psi_i \in \mbH$ as the eigenvalues and corresponding eigenfunctions of $C$. According to the \textit{KL-expansion} theorem
\begin{align*}
\norm{X}_{\mbH}^2=  \mathlarger{\sum\limits_{j=1}^{\infty}} \langle X , \psi_j \rangle_{\mbH}^2 =  \mathlarger{\sum\limits_{j=1}^{\infty}} \gamma_j \langle X , \dfrac{\psi_j}{\sqrt{\gamma_j}} \rangle_{\mbH}^2 = \mathlarger{\sum\limits_{j=1}^{\infty}} \gamma_j Z_j^2,
\end{align*}
where $Z_j$ is a subgaussian process in $\mbR$ \citep{Antonioni1997sub} with parameter $\left\langle \dfrac{\psi_j}{\sqrt{\gamma_j}}  , \dfrac{C(\psi_j)}{\sqrt{\gamma_j}} \right\rangle = 1$. Define the events
\begin{align*}
A_J=\bigg\{ \mathlarger{\sum\limits_{j=1}^{J}} \gamma_j Z_j^2 \geq \| C\|_{1} + 2 \|C \|_{2} \sqrt{t} + 2 \| C\|_{\infty} t \bigg\} \ \ \ \ \ \ \ J = 1 ,2 , \dots
\end{align*}
Since $\norm{C}_1 \geq \sum\limits_{i=1}^{J} \gamma_{i}$ and $\norm{C}_2^2 \geq \sum\limits_{i=1}^{J} \gamma_{i}^2$, based on \citep{hsu2012tail} we can see
\begin{align*}
P(A_J) \leq P \left( \sum\limits_{j=1}^{J} \gamma_j Z_j^2 \geq \sum\limits_{i=1}^{J} \gamma_{i}  + 2 \sqrt{\sum\limits_{i=1}^{J} \gamma_{i}^2} \sqrt{t} + 2 \| C\|_{\infty} t \right) \leq e^{-t}.
\end{align*}
Since $A_1 \subset A_2 \subset \hdots $ and using continuity from below, we can conclude
\begin{align*}
P\left( \sum\limits_{j=1}^{\infty} \gamma_j Z_j^2 \geq \| C\|_{1} + 2 \|C \|_{2} \sqrt{t} + 2 \| C\|_{\infty} t \right) \leq e^{-t}.
\end{align*}
\end{proof}
\end{lemma}
\begin{lemma} \label{l:idempotent}
If $Q$ is an operator in $\mbH$ such that $\norm{Q^2x} \leq \norm{Qx}$ for any $x \in \mbH$, the eigenvalues of $Q$ will be in $[0,1]$.
\end{lemma}
\begin{proof}
let's denote $\theta_j$ and $v_j$ as the eigenvalues and eigenfunctions of $Q$. So $Qv_i=\theta_i v_i$ and then
\begin{align*}
Q^2 v_i = \theta_i Q v_i = \theta_i^2 v_i.
\end{align*}
Then
\begin{align*}
\theta_i^2 \norm{v_i} = \norm{Q^2 v_i} \leq \norm{Qv_i}=|\theta_i| \norm{v_i},
\end{align*}
or equivalently
\begin{align*}
\theta_i^2 \leq |\theta_i|,
\end{align*}
which implies $ 0 \leq \theta_i \leq 1$.
\end{proof}
\begin{lemma} \label{l:PCP}
Assume that $Q$ is a continuous linear operator and $C$ a covariance operator over $\mbH$.  For an arbitrary covariance operator, $A$, let $\| A \|_m$ denote the $m$-norm of the eigenvalues of $A$.  Then we have that  
\begin{align*}
\norm{QCQ^{\star}}_m \leq \norm{Q}_{op}^2 \norm{C}_m \ \ \ \ \ \forall m=1,\dots,\infty,
\end{align*}
where $\norm{Q}_{op}$ is the operator norm of $Q$, equivalently the largest singular value of $Q$.
\end{lemma}
\begin{proof}
Let's define $\theta^{\prime}_i$ as eigenvalues of $QCQ^{\star}$, then we have
\begin{align*}
\theta^{\prime}_i =\sup_{x \in \mbH_i} \langle QCQ^{\star}(x) , x \rangle_{\mbH},
\end{align*}
where $\mbH_i=\{ x ; \ \ \norm{x}_{\mbH}=1  \ \ \& \ \  \langle x , v_j \rangle_{\mbH} =0 \ \ \ \forall j=1,\dots, i-1\}$. Since $C$ is a covariance operator, it is self-adjoint with positive eigenvalues, then we have \\
\begin{align*}
\theta^{\prime}_i & =\sup\limits_{x \in \mbH_i} \left\langle QCQ^{\star}(x) , x \right\rangle_{\mbH} = \sup\limits_{x \in \mbH_i} \left\langle QC^{\haf} (QC^{\haf})^{\star}(x) , x \right\rangle_{\mbH} =  \sup\limits_{x \in \mbH_i} \left\langle (QC^{\haf})^{\star}(x) , (QC^{\haf})^{\star}(x) \right\rangle_{\mbH} \\ \\ 
&= \sup\limits_{x \in \mbH_i} \norm{(QC^{\haf})^{\star}(x)}_{\mbH}^2 = \sup\limits_{x \in \mbH_i} \norm{QC^{\haf}(x)}_{\mbH}^2 \leq \norm{Q}_{op}^2 \sup\limits_{x \in \mbH_i} \norm{C^{\haf}(x)}^2_{\mbH} \\ \\ 
&= \norm{Q}_{op}^2  \sup\limits_{x \in \mbH_i} \left\langle C^{\haf}(x) , C^{\haf}(x) \right\rangle_{\mbH} =  \norm{Q}_{op}^2  \sup\limits_{x \in \mbH_i} \left\langle C(x) , x \right\rangle_{\mbH} = \norm{Q}_{op}^2 \theta_i. \\
\end{align*}
So based on $\theta_i^{\prime} \leq \norm{Q}_{op}^2 \theta_i$, we can conclude
\begin{align*}
\left(\sum\limits_{i=1}^{\infty} {\theta_i^{\prime}}^m\right)^{\frac{1}{m}} \leq \norm{Q}_{op}^2 \left(\sum\limits_{i =1}^{\infty} \theta_i^m\right)^{\frac{1}{m}},
\end{align*}
then we will have
\begin{align*}
\norm{QCQ^{\star}}_m \leq \norm{Q}_{op}^2 \norm{C}_m \hspace{1cm} \forall m=1,\dots,\infty.
\end{align*}
\end{proof}
Finally for some technical proofs, we recall the following lemma from \citet{barber2017function}.
\begin{lemma} \label{l:btilde}
If Assumption \ref{assumption:4p} holds, the FSL estimate $\tilde{\beta_i}$ and true coefficient $\beta_i^{\star}$ satisfy
\begin{align*}
\sup\limits_{i \in \mcS} \norm{\tbi - \beta^{\star}_i}_{\mbH}=O_{p}(r_N^{\haf}),
\end{align*}
where $r_N=\dfrac{I_0 \log(I)}{N}$.
\end{lemma}
\section*{Proof of the Lemma \ref{l:hbeta}}
Let's fix an $i \in \{ 1,\dots,I\}$. We want to find the $\hat \beta_i$ which minimizes the target function \eqref{AFSSEN}. The idea is same as Lemma \ref{l:hbeta2} for a univariate case.
Lets denote $\chb_{i} = \dfrac{1}{N} \mathlarger{\sum\limits}_{n=1}^{N} X_{ni} ( Y_{n} - \mathlarger{\sum\limits}_{j \neq i} X_{nj} \hat\beta_{j})$.
According to the Lemma \ref{l:subdiff}, when $\hbi=0$
\begin{align*}
-K(\chb_{i})+\laa\tw_{i}h=0 \hspace{20mm}  \{ h ; \| K^{-\frac{1}{2}}(h) \|_{\mbK} \leq 1 \}, \\
K^{\frac{1}{2}}(\chb_{i})=\laa \tw_{i} K^{-\frac{1}{2}}(h) \hspace{20mm}  \{ h ; \| K^{-\frac{1}{2}}(h) \|_{\mbK} \leq 1\}.
\end{align*}
So the above equality exists when $\norm{K^{\frac{1}{2}}(\chb_{i})}_{\mbK} \leq \laa \tw_{i}$ or equivalently $ \norm{\chb_{i}}_{\mbH} \leq \laa \tw_{i}$. \\
In the other side, when $ \norm{\chb_{i}}_{\mbH} \geq \laa \tw_{i}$, we have
\begin{align}
& -K(\chb_{i})+K(\hbi)+ \la L^2(\hbi) + \laa \tw_{i} K(\hat{\beta}_i)\norm{\hat{\beta}_i}_{\mbH}^{-1} = 0, \nonumber \\ 
& \left(K+ \la L^2 +\frac{\laa\tw_{i}K}{\norm{\hat{\beta}_i}_{\mbH}} \right) \hbi=K(\chb_{i}), \nonumber \\ 
& \hbi=\left((1 +\frac{\laa\tw_{i}}{\norm{\hat{\beta}_i}_{\mbH}})K+ \la L^2 \right)^{-1} K(\chb_{i}), \nonumber \\
& \hbi=\left((1+\frac{\laa\tw_{i}}{\norm{\hat{\beta}_i}_{\mbH}}) \mathds{I} + \la K^{-1} L^2  \right)^{-1} \chb_{i}. \label{l:numeq}
\end{align} 
\section*{Proof of the Lemma \ref{l:num}}
Taking $\norm{.}_{\mbH}^{2}$ from the both hand side of equation \eqref{l:numeq} 
\begin{align*}
\norm{\hat{\beta}_i}_{\mbH}^2 & =\norm{\left((1+\frac{\laa\tw_{i}}{\norm{\hat{\beta}_i}_{\mbH}}) \mathds{I} + \la K^{-1} L^2  \right)^{-1} \chb_{i}}_{\mbH}^2. \\
\end{align*}
For ease of notation, we use $A=\left((1+\dfrac{\laa\tw_{i}}{\norm{\hat{\beta}_i}_{\mbH}}) \mathds{I} + \la K^{-1} L^2  \right)^{-1}$ for the following parts.
\begin{align*}
 \norm{\hat{\beta}_i}_{\mbH}^2 & =\norm{A \chb_i}_{\mbH}^2 =\sum_{j=1}^{\infty} \langle A\chb_{i} , v_{j} \rangle_{\mbH}^2 = \sum_{j=1}^{\infty} \langle \chb_{i} , Av_{j} \rangle_{\mbH}^2, \\ \\
& = \sum_{j=1}^{\infty} \langle \chb_{i} , \frac{1}{\left(1+\dfrac{\laa\tw_{i}}{\norm{\hat{\beta}_i}_{\mbH}}\right) + \la \eta_j^2 \theta_{j}^{-1}} v_{j} \rangle_{\mbH}^2, \\ \\
& = \sum_{j=1}^{\infty} \frac{\langle \chb_{i} , v_{j} \rangle_{\mbH}^2 } { \left( (1+\dfrac{\laa\tw_{i}}{\norm{\hat{\beta}_i}_{\mbH}}) + \la \eta_j^2 \theta_{j}^{-1} \right)^2},
\end{align*}
or equivalently
\begin{align*}
1 = \sum_{j=1}^{\infty} \frac{ \langle \chb_{i} , v_{j} \rangle_{\mbH}^2 }{\left( (1 +\la\eta_j^2 \theta_j^{-1} )\norm{\hat{\beta}_i}_{\mbH} + \laa \tw_{i}\right)^2}.
\end{align*}
\section*{Proof of Theorem \eqref{t:w.oracle}}
\textbf{{part 1:}} \\
One can see
\begin{align} \label{l:support}
\hat{\mcS} = \mcS \longleftrightarrow \left\{
                \begin{array}{ll}
                 \hbi \neq 0 \qquad & \quad \forall i \in \mcS \\
                  \\ \\
\hbi=0 \qquad & \quad\forall i \not\in \mcS \\
                \end{array}
\right.
\end{align}
where the $\hat{\mcS}$ and $\mcS$ are support of the estimated AFSSEN and true predictors respectively. So we can see \eqref{l:support} can be induced from\\
\begin{align*}
\left\{
                \begin{array}{ll}
                 \| \sbi - \hbi  \|_{\mbH} <  \| \sbi \|_{\mbH} \qquad & \quad \forall i \in \mcS \\
                  \\ \\
\norm{\dfrac{1}{N} \bX_{.i}^\top (\bY-\bX_{1}\hboldb_{1})}_{\mbH} \leq \laa\tw_{i} \qquad & \quad\forall i \not\in \mcS \\
                \end{array}
\right.
\end{align*} 
or equivalently
\begin{align*}
\left\{
                \begin{array}{ll}
                 \| \eit ( \stb_{1} - \hboldb_{1} ) \|_{\mbH}  <  \| \sbi \|_{\mbH} \qquad & \quad \forall i \in \mcS \\
                  \\ \\
\dfrac{1}{N} \norm{ \bX_{.i}^{\top} (\bY - \bX_{1} \hboldb_{1}) }_{\mbH} \leq \laa\tw_{i} \qquad & \quad\forall i \not\in \mcS \\
                \end{array}
\right.
\end{align*} 
According to Lemma \ref{l:hbeta2} we have
\begin{align*}
\hboldb_{1} & = \hsig G_{I_0} (\stb_{1}) + G_{I_0} \left(  N^{-1} \bX_{1}^\top \ep - \laa \sbb \right), \\
\hboldb_{1} - \stb_{1} & =  \left( \hsig G_{I_0} - \mathds{I}_{I_0} \right) \stb_{1} + G_{I_0} \left( N^{-1} \bX_{1}^\top \ep - \laa \sbb \right) \\
& =  G_{I_0} \left( \hsig \mathds{I}_{I_0} - G_{I_0}^{-1} \right) \stb_{1} + G_{I_0} \left( N^{-1} \bX_{1}^\top \ep - \laa \sbb \right) \\
& =  G_{I_0} \left( \hsig \mathds{I}_{I_0} - \left(\hsig \mathds{I}_{I_0}+\la K_{I_0}^{-1} L_{I_0}^2 \right) \right) \stb_{1} + G_{I_0} \left( N^{-1} \bX_{1}^\top \ep - \laa \sbb \right) \\
& = - \la G_{I_0} K_{I_0}^{-1} L^2_{I_0}  (\stb_{1}) + N^{-1} G_{I_0}(\bX_{1}^\top \ep) - \laa G_{I_0}(\sbb),
\end{align*} 
where $G_{I_0} = \left(\hsig \mathds{I}_{I_0}+\la K_{I_0}^{-1} L_{I_0}^2 \right)^{-1}$. In some sense, $ -\la G_{I_0} K_{I_0}^{-1} L^2_{I_0} (\stb_{1})$ and $ N^{-1} G_{I_0}(\bX_{1}^\top \ep) - \laa G_{I_0}(\sbb)$ play the role of \textit{Bias} and \textit{Variance} respectively.
Since $\hbi=0$ for all $i \not\in S$, we can see
\begin{align*}
& \dfrac{1}{N} \bX_{.i}^{\top} (\bY - \bX_{1} \hboldb_{1}) = \frac{1}{N} \bX_{.i}^{\top}  \left( \bX_{1} (\stb_{1} - \hboldb_{1} ) + \ep \right) \\
& = \frac{1}{N} \bX_{.i}^{\top}  \left[ \bX_{1} \left( \la G_{I_0} K_{I_0}^{-1} L^2_{I_0} (\stb_{1}) + \laa G_{I_0}(\sbb) - N^{-1} G_{I_0}(\bX_{1}^\top \ep) \right) + \ep \right]  \\
& = \frac{1}{N} \bX_{.i}^{\top} \left[ \la \bX_1 G_{I_0} K_{I_0}^{-1} L^2_{I_0} (\stb_{1}) + \laa \bX_1 G_{I_0}(\sbb) + \left(\ep - N^{-1} \bX_1 G_{I_0}(\bX_1^\top \ep) \right) \right] \\
& = \frac{1}{N} \bX_{.i}^{\top} \left[ \la \bX_1 G_{I_0} K_{I_0}^{-1} L^2_{I_0} (\stb_{1}) + \laa \bX_1 G_{I_0}(\sbb) + \left(\mathds{I}_{N} - \bX_1 \left(\bX_{1}^\top \bX_{1}  \mathds{I}_{I_0} + \la N K_{I_0}^{-1} L^2_{I_0} \right)^{-1} \bX_1^\top \right) \ep \right] \\
& = \frac{1}{N} \bX_{.i}^{\top} \left[ \la \bX_1 G_{I_0} K_{I_0}^{-1} L^2_{I_0} (\stb_{1}) + \laa \bX_1 G_{I_0}(\sbb) + H_N \ep \right],
\end{align*}
where
\begin{align*}
H_N=\left(\mathds{I}_{N} - \bX_1 \left(\bX_{1}^\top \bX_{1}  \mathds{I}_{I_0} + \la N K_{I_0}^{-1} L^2_{I_0} \right)^{-1} \bX_1^\top \right).
\end{align*}
So with using above achievements we can see \eqref{l:support} is equivalent to
\begin{align*}
\left\{
\begin{array}{ll}
    \norm{ \eit \left( - \la G_{I_0} K_{I_0}^{-1} L^2_{I_0}  (\stb_{1}) + N^{-1} G_{I_0}(\bX_{1}^\top \ep) - \laa G_{I_0}(\sbb) \right)}_{\mbH} < \norm{\sbi}_{\mbH} \qquad & \quad \forall i \in \mcS \\
                  \\ \\
\norm{ \frac{1}{N} \bX_{.i}^{\top} \left[ \la \bX_1 G_{I_0} K_{I_0}^{-1} L^2_{I_0} (\stb_{1}) + \laa \bX_1 G_{I_0}(\sbb) + H_N \ep \right] }_{\mbH} \leq \laa\tw_{i} \qquad & \quad\forall i \not\in \mcS \\
                \end{array}
\right.
\end{align*}
It is easy to see that $\{\hat{\mcS} \neq \mcS\} \subseteq \bigcup\limits_{i=1}^{6} B_{i}$ where
\begin{align*}
B_{1} & =\bigg\{\la \norm{ \eit G_{I_0} K_{I_0}^{-1} L^2_{I_0}  (\stb_{1})}_{\mbH} \geq \dfrac{\norm{\sbi}_{\mbH}}{3} & i \in \mcS \bigg\}, \\
B_{2} & =\bigg\{\frac{1}{N} \norm{ \eit G_{I_0}(\bX_{1}^\top \ep)}_{\mbH} \geq \dfrac{\norm{\sbi}_{\mbH}}{3}  & i \in \mcS \bigg\}, \\
B_{3} & =\bigg\{\laa \norm{ \eit G_{I_0}(\sbb) }_{\mbH}\geq \dfrac{\norm{\sbi}_{\mbH}}{3} & i \in \mcS \bigg\}, \\
B_{4} & =\bigg\{\frac{\la}{N} \norm{ \bX_{.i}^\top \bX_1 G_{I_0} K_{I_0}^{-1} L^2_{I_0} (\stb_{1})}_{\mbH} \geq \dfrac{\laa\tw_{i}}{3} & i \not\in \mcS \bigg\}, \\
B_{5} & =\bigg\{\frac{1}{N} \norm{\bX_{.i}^\top  \bX_1 G_{I_0}(\sbb)}_{\mbH} \geq \dfrac{\tw_{i}}{3}  & i \not\in \mcS \bigg\}, \\
B_{6} & =\bigg\{\frac{1}{N} \norm{\bX_{.i}^\top H_N \ep}_{\mbH} \geq \dfrac{\laa\tw_{i}}{3} & i \not\in \mcS \bigg\}. \\
\end{align*} 
So for proving Theorem \ref{t:w.oracle}, we just need to show that $P(B_{i})$ asymptotically goes to zero for all $i = 1, \dots, 6$. \\ \\
\textbf{Step 1:} $P(B_{1}) \rightarrow 0$ \\
Recall that
\begin{align*}
B_{1} & =\bigg\{\la \norm{ \eit G_{I_0} K_{I_0}^{-1} L^2_{I_0}  (\stb_{1})}_{\mbH} \geq \dfrac{\norm{\sbi}_{\mbH}}{3} & i \in \mcS \bigg\},
\end{align*}
where $G_{I_0} = \left(\hsig \mathds{I}_{I_0}+\la K_{I_0}^{-1} L_{I_0}^2 \right)^{-1}$. We aim to show that $\dfrac{\la \norm{ \eit G_{I_0} K_{I_0}^{-1} L^2_{I_0}  (\stb_{1})}_{\mbH}}{\norm{\sbi}_{\mbH}} \rightarrow 0 $. We can write
\begin{align}
& \frac{\la \norm{ \eit G_{I_0} K_{I_0}^{-1} L^2_{I_0}  (\stb_{1})}_{\mbH}}{\norm{\sbi}_{\mbH}} = \frac{\la \norm{ \eit G_{I_0} K_{I_0}^{-1/2} L^2_{I_0} \dfrac{K_{I_0}^{-1/2}(\stb_{1})}{\norm{K_{I_0}^{-1/2}(\stb_{1})}_{\mbH^{I_0}}}}_{\mbH} \norm{K_{I_0}^{-1/2}(\stb_{1})}_{\mbH^{I_0}}}{\norm{\sbi}_{\mbH}}  \nonumber \\ \nonumber \\
&\leq \frac{\la \norm{\eit G_{I_0} K_{I_0}^{-1/2} L^2_{I_0}}_{op} \norm{K_{I_0}^{-1/2}(\stb_{1})}_{\mbH^{I_0}} }{\norm{\sbi}_{\mbH}} \leq \frac{ \la \norm{\eit G_{I_0} K_{I_0}^{-1/2} L^2_{I_0}}_{op} \norm{\stb_{1}}_{\mbK^{I_0}} }{ \min\limits_{i \in \mcS}\norm{\sbi}_{\mbH}} \nonumber \\ \nonumber
& \leq \frac{ I_{0}^{\nicefrac{1}{2}} \max\limits_{i \in \mcS} \norm{\sbi}_{\mbK} }{ \min\limits_{i \in \mcS}\norm{\sbi}_{\mbH}} \la \norm{ \eit G_{I_0} K_{I_0}^{-1/2} L^2_{I_0}}_{op} \leq \dfrac{ I_0^{\nicefrac{1}{2}} d_{N}}{b_{N}} \la \norm{ \eit G_{I_0} K_{I_0}^{-1/2} L^2_{I_0}}_{op} \nonumber \\
& \leq \dfrac{ I_0^{\nicefrac{1}{2}} d_{N}}{b_{N}} \la \norm{\eit} \norm{ G_{I_0} K_{I_0}^{-1/2} L^2_{I_0}}_{op} = \dfrac{ I_0^{\nicefrac{1}{2}} d_{N}}{b_{N}} \la \norm{ G_{I_0} K_{I_0}^{-1/2} L^2_{I_0}}_{op},  \label{p:B1.1}
\end{align}
\\
where $d_{N}= \max\limits_{i \in \mcS}\norm{\sbi}_{\mbK} $ and $b_{N}=\min\limits_{i \in \mcS}\norm{\sbi}_{\mbH}$. So We need to find an upper bound for $ \norm{G_{I_0} K_{I_0}^{-1/2} L^2_{I_0}}_{op}$ which is  the maximum eigenvalue of $G_{I_0} K_{I_0}^{-1/2} L^2_{I_0}$. According to the tensor product definition \citep{Kokoszka2017Introduction}, we have
\begin{align*}
G_{I_0} K_{I_0}^{-1/2} L^2_{I_0} & = \left(\hsig \mathds{I}_{I_0}+\la K_{I_0}^{-1} L_{I_0}^2 \right)^{-1} K_{I_0}^{-1/2} L^2_{I_0} \\
& = \left(\hsig K_{I_0}^{1/2} L^{-2}_{I_0} + \la K_{I_0}^{-1/2} \right)^{-1} \\
& = (\hsig \otimes K^{1/2}L^{-2} + \la \bI_{I_0} \otimes K^{-1/2})^{-1}.
\end{align*}
where $\bI_{I_0}$ is an identity $I_0$ by $I_0$ matrix. In order to find the eigenvalues of $G_{I_0} K_{I_0}^{-1/2} L^2_{I_0}$, let's denote $u_i$ as the eigenfunction of $\hsig$ and $v_j$ as the eigenfunctions of $K$ and $L$. Then
\begin{align*}
(\hsig \otimes K^{1/2}L^{-2} + \la \bI_{I_0} \otimes K^{-1/2})(u_i \otimes v_j) & = (\tau_i \theta_j^{1/2}\eta_j^{-2} + \la \theta_j^{-1/2})^{-1}(u_i \otimes v_j),
\end{align*}
where $\tau_i$ is the eigenvalue of $\hsig$ and then $(\tau_i \theta_j^{1/2}\eta_j^{-2} + \la \theta_j^{-1/2})^{-1}$ can be considered as the eigenvalues of 
 $G_{I_0} K_{I_0}^{-1/2} L^2_{I_0}$.
Since
\begin{align} \label{t:GKL}
\dfrac{1}{\tau_i \theta_j^{1/2}\eta_j^{-2} + \la \theta_j^{-1/2}} \leq \dfrac{1}{\tau^{-1} \eta_1^{-2} \theta_j^{1/2} + \la \theta_j^{-1/2}},
\end{align}
the maximum value of \eqref{t:GKL} occures in $\theta_j = \la \tau \eta_1^2$. Then
\begin{align} \label{s:GK1/2L}
\norm{G_{I_0} K_{I_0}^{-1/2} L^2_{I_0}}_{op} \leq \dfrac{\eta_1 \sqrt{\tau}}{2 \sqrt{\la}}.
\end{align}
So we can conclude
\begin{align*}
\dfrac{ I_0^{\nicefrac{1}{2}} d_{N}}{b_{N}} \la \norm{ G_{I_0} K_{I_0}^{-1/2} L^2_{I_0}}_{op} \leq \dfrac{ I_0^{\nicefrac{1}{2}} d_{N}}{b_{N}}\dfrac{\eta_1 \sqrt{\la \tau}}{2}.
\end{align*}

If we assume $\la\ll \dfrac{b_N^2}{d_N^2 I_0}$, we can easily see $P(B_1) \rightarrow 0$ asymptotically. \\ \\
\textbf{Step 2:} $P(B_{2}) \rightarrow 0$ \\
Recall that 
\begin{align*}
B_{2}  =\bigg\{\frac{1}{N} \norm{ e_{i}^\top  G_{I_0}(\bX_{1}^\top \ep)}_{\mbH} \geq \dfrac{\norm{\sbi}_{\mbH}}{3}  \hspace{1cm}  ;i \in \mcS \bigg\},
\end{align*}
where $G_{I_0} = \left(\hsig \mathds{I}_{I_0}+\la K_{I_0}^{-1} L_{I_0}^2 \right)^{-1}$.
We notice $B_{1} = \bigcup\limits_{i \in \mcS} A_{i}$ such that
\begin{align*}
A_{i} & =\bigg\{\frac{1}{N} \norm{e_{i}^\top  G_{I_0}(\bX_{1}^\top \ep)}_{\mbH} \geq \dfrac{\norm{\sbi}_{\mbH}}{3}  \bigg\} \\
& = \bigg\{\norm{ Q_i(\bX_{1}^\top \ep)}_{\mbH}^2 \geq \dfrac{\norm{\sbi}_{\mbH}^2}{9}  \bigg\},
\end{align*}
where 
\begin{align} \label{p:Qi}
Q_i = N^{-1} e_i^{\top} G_{I_0}
\end{align}
is a continuous linear operator from $\mbH^{I_0}$ to $\mbH$. Then we can see
\begin{align} \label{B1:1}
P(B_{1}) \leq \sum\limits_{i \in \mcS} P(A_{i}) = \sum\limits_{i \in \mcS} P\left( \norm{Q_i (\bX_{1}^\top \ep)}_{\mbH}^2 \geq \dfrac{\norm{\sbi}_{\mbH}^2}{9} \right) \leq  \sum\limits_{i \in \mcS} P\left( \norm{Q_i (\bX_{1}^\top \ep)}_{\mbH}^2 \geq \dfrac{b_{N}^2}{9}\right),
\end{align}
where $b_{N}=\min\limits_{i \in \mcS}\norm{\sbi}_{\mbH}$.
So we just need to find an upper bound for the right hand side of \eqref{B1:1}. \\
Since $\ep=(\epsilon_1,\dots,\epsilon_N) \in \mbH^N$ where $\epsilon_i$ are independent mean zero $C$-subgaussian process in $\mbH$, then $\bX_{1}^{\top} \ep$ will be a $C_{1}$-subgaussian in $\mbH^{I_0}$ such that
\begin{align*}
C_1 = \bX_1^{\top} \bX_1 C_{I_0} = N \hsig C_{I_0},
\end{align*}
where, since $C_{I_0}$ is applied coordinate wise, $\hat \Sigma_{11}$ and $C_{I_0}$ are interchangeable and thus this is a valid covariance matrix.  Based on an extension of Lemma \ref{l:multi-sub} in $\mbH^{I_0}$, $Q_i (\bX_1^{\top} \ep)$ will be a $C_q$-subgaussian process with
\begin{align*}
C_q = (Q_i)(N \hsig C_{I_0})(Q_i^{\top}).
\end{align*}
According to Lemma \ref{l:subg-Basic} we have
\begin{align*}
P\left(\norm{Q_i \bX_1^{\top} \ep}_{\mbH}^2 \geq  \norm{C_{q}}_{1} + 2 \norm{C_{q}}_{2} \sqrt{t} + 2 \norm{C_{q}}_{\infty} t \right) \leq e^{-t}.
\end{align*}
Then based on Lemma \ref{l:PCP} and \eqref{p:Qi}
\begin{align*}
\norm{C_q}_m \leq \norm{Q_i}_{op}^2 \norm{N \hsig}_{op} \norm{C}_{m} \leq N^{-1} \norm{G_{I_0}}_{op}^2 \tau \norm{C}_m.
\end{align*}
So we are going find an upper bound for $\norm{G_{I_0}}_{op}$. Using tensor product notation as in step 1, we have
\begin{align*}
G_{I_0} = \left(\hsig \mathds{I}_{I_0}+\la K_{I_0}^{-1} L_{I_0}^2 \right)^{-1} = (\hsig \otimes \mathds{I} + \la \bI_{I_0} \otimes K^{-1} L^2)^{-1}.
\end{align*}
where $\mathds{I}$ is an identity operator from $\mbH$ to $\mbH$ and $\bI_{I_0}$ is an identity $I_0$ by $I_0$ matrix. Now we can write the eigenvalues of $G_{I_0}$ as
\begin{align*}
G_{I_0} (u_i \otimes v_j) = \left( \tau_i + \la \theta_j^{-1} \eta_j^2 \right)^{-1} (u_i \otimes v_j).
\end{align*}
According to Assumption \ref{a:1.3}
\begin{align}  \label{t:GI0}
\norm{G_{I_0}}_{op} = \max  \left( \tau_i + \la \theta_j^{-1} \eta_j^2 \right)^{-1} = \dfrac{1}{\tau_i}\max \dfrac{\tau_i \theta_j}{\tau_i \theta_j + \la \eta_j^{2}} \leq \dfrac{1}{\tau_i} \leq \tau,
\end{align}
then we can conclude
\begin{align*}
\norm{C_q}_m \leq N^{-1} \tau^3 \norm{C}_m,
\end{align*}
and finally we will have
\begin{align*}
P\left(\norm{Q_i \bX_1^{\top} \ep}_{\mbH}^2 \geq \dfrac{\tau^3}{N} ( \norm{C}_{1} + 2 \norm{C}_{2} \sqrt{t} + 2 \norm{C}_{\infty} t) \right) \leq e^{-t}.
\end{align*}
So we are looking to find a $\hat{t}$ such that
\begin{align*}
\frac{b_{N}^2}{9} \geq \frac{\tau^3}{N} \left( \norm{C}_{1} + 2 \norm{C}_{2} \sqrt{\hat{t}} + 2 \norm{C}_{\infty} \hat{t} \right).
\end{align*}
Since $C$ is a covariance operator, its nuclear property will implify there exists a constant D which 
\begin{align} \label{l:CDT}
( \norm{C}_{1} + 2 \norm{C}_{2} \sqrt{t} + 2 \norm{C}_{\infty} t) \leq Dt,
\end{align}
then
\begin{align*}
\frac{b_{N}^2}{9} \geq \frac{ \tau^3}{N} D \hat{t}.
\end{align*}
We Choose $\hat{t} = \dfrac{N b_{N}^2}{9\tau^3 D } $. Therefore  \eqref{B1:1} will be written as
\begin{align} \label{B1:2}
P(B_{1}) \leq \sum\limits_{i \in \mcS} \exp\left(- \dfrac{N b_{N}^2}{9\tau^3 D }\right) \leq  \exp\left(- \dfrac{N b_{N}^2}{9\tau^3 D } + \log (I_{0})\right).
\end{align}
According to Assumption \ref{assumption:4p},
the right hand side of \eqref{B1:2} goes to zero because $N b_N^2 \to \infty$ and  
\begin{align*}
N b_{N}^2 \gg I_{0}^2 \log(I)  
\Longrightarrow N b_N^2 \gg \log(I_0).
\end{align*} \\ \\ 
\textbf{step 3:} $P(B_{3}) \rightarrow 0$ \\
Recall that
\begin{align} \label{B3:1}
B_{3} =\bigg\{\laa \norm{ \eit G_{I_0}(\sbb)}_{\mbH}\geq \dfrac{\norm{\sbi}_{\mbH}}{3} \hspace{1cm} ; \ i \in \mcS \bigg\},
\end{align}
where $\sbb=\{ \tilde{w}_{i} \hbi \norm{\hbi}_{\mbH}^{-1} ; i \in \mcS \}$. Our aim is to show that $\dfrac{\laa \norm{\eit G_{I_0}(\sbb)}_{\mbH}}{ \norm{\sbi}_{\mbH}} \rightarrow 0 $.
By using the Assumption \ref{assumption:4p} we can write
\begin{align} \label{p:B3}
\frac{\laa \norm{ \eit G_{I_0}(\sbb) }_{\mbH}}{\norm{\sbi}_{\mbH} } \leq \frac{\laa \norm{ \eit G_{I_0}\dfrac{\sbb}{\norm{\sbb}_{\mbH^{I_0}}}}_{\mbH} \norm{\sbb}_{\mbH^{I_0}}}{\min\limits_{i \in \mcS}\norm{\sbi}_{\mbH} } \leq  \frac{\laa \norm{G_{I_0}}_{op} \norm{\sbb}_{\mbH^{I_0}}}{b_N }.
\end{align}
So we need to find the upper bounds of $ \norm{ G_{I_0}}_{op}$ and $\norm{\sbb}_{\mbH^{I_0}}$. \\ \\
First, same as what we did in \eqref{t:GI0}, we have
\begin{align} \label{p:B3.GK}
\norm{G_{I_0}}_{op} \leq \tau.
\end{align}
Second, we can write
\begin{align*}
\norm{\sbb}_{\mbH^{I_0}}^2 =\sumiS \tw_i^2  \norm{\hbi}_{\mbH}^{-2} \norm{\hbi}_{\mbH}^2
= \sumiS w_i^2 + \sumiS (\tw_i^2 - w_i^2),
\end{align*}
where $\tw_i=\norm{\tilde{\beta}_i}_{\mbH}^{-1}$ and 
$w_i=\norm{\sbi}_{\mbH}^{-1}$. \\ \\
With using \textit{Taylor Expansion} for functional data $f(x+h)-f(x)= \langle h , f^{\prime}(x) \rangle + o(h^2)$ where  $f(x)=\dfrac{1}{\norm{x}^2}$ and $f^{\prime}(x) = \dfrac{-2x}{ \norm{x}^4}$, we can write
\begin{align*}
\tw_i^2 - w_i^2 = \frac{1}{\norm{\tilde{\beta}_i}_{\mbH}^2}-\frac{1}{\norm{\sbi}_{\mbH}^2} \approx \left\langle \tbi - \sbi , \frac{-2}{\norm{\sbi}_{\mbH}^4} \sbi \right\rangle_{\mbH}.
\end{align*}
According to \textit{Cauchy-Schwarz} inequality and Lemma \ref{l:btilde}
\begin{align*}
& \tw_i^2 - w_i^2 = \frac{2}{\norm{\sbi}_{\mbH}^4}\left\langle \tbi - \sbi ,  \sbi \right\rangle_{\mbH} \leq \frac{2}{\norm{\sbi}_{\mbH}^3}\norm{ \tbi - \sbi }_{\mbH} \\ \\
& \leq \frac{2}{b_{N}} \norm{\tbi-\sbi}_{\mbH} \frac{1}{\norm{\sbi}_{\mbH}^2} \leq 
\frac{2}{b_{N}} \left( \sup\norm{\tbi-\sbi}_{\mbH} \right) \frac{1}{\norm{\sbi}_{\mbH}^2} \\ \\
& = \frac{2}{ b_{N}} O_{p}(r_{N}^{\frac{1}{2}}) w_{i}^2,
\end{align*}
where $r_N = \dfrac{I_0 \log{(I)}}{N}$. 
By using Assumption \ref{assumption:4p} we can see $\dfrac{r_N^{\haf}}{b_N} \rightarrow 0$ and then
\begin{align} \label{B3:shat2}
\norm{\sbb}_{\mbH^{I_0}}^2 \leq \left(\sumiS w_{i}^2 \right) \left(1+\frac{ 2 O_{p}(r_{N}^{\frac{1}{2}})}{b_{N}}\right)= \left(\sumiS w_{i}^2 \right) \left(1 + o_p(1) \right) \leq \frac{I_{0}}{b_{N}^2} \left(1+o_{p}(1) \right).
\end{align}

According to \ref{p:B3}, we can conclude
\begin{align*}
\dfrac{\laa \norm{\eit G_{I_0} \sbb}_{\mbH}}{\norm{\sbi}_{\mbH} } \leq \frac{\laa \tau \sqrt{I_{0}} \left(1+o_{p}(1) \right)^{\haf} }{b_N^2 } \rightarrow 0,
\end{align*}
so if $\laa \ll \dfrac{b_N^2}{\sqrt{I_0}}$, then $P(B_3) \rightarrow 0$. \\ \\ 
\textbf{step 4:} $P(B_{4}) \rightarrow 0$ \\
The idea is same as Part 1. Recall that
\begin{align*}
B_{4} & =\bigg\{\frac{\la}{N} \norm{ \bX_{.i}^\top \bX_1 G_{I_0} K_{I_0}^{-1} L^2_{I_0} (\stb_{1})}_{\mbH} \geq \dfrac{\laa\tw_{i}}{3} \hspace{1cm} ;  i \not\in \mcS \bigg\},
\end{align*}
where $G_{I_0} = \left(\hsig \mathds{I}_{I_0}+\la K_{I_0}^{-1} L_{I_0}^2 \right)^{-1}$. We aim to show $ \dfrac{\la \norm{ \bX_{.i}^\top \bX_1 G_{I_0} K_{I_0}^{-1} L^2_{I_0} (\stb_{1})}_{\mbH}}{N \laa\tw_{i}  } \rightarrow 0 $.
According to Lemma \ref{l:btilde}, Assumption \ref{a:1.4} and equation \eqref{s:GK1/2L} we have \\
\begin{align*}
& \frac{ \la \norm{ \bX_{.i}^\top \bX_1 G_{I_0} K_{I_0}^{-1} L^2_{I_0} (\stb_{1})}_{\mbH}}{ N \laa\tw_{i}} \leq 
\frac{ \la \norm{\hsigg \hsig^{-1} }_{op} \norm{\hsig}_{op} \norm{ G_{I_0} K_{I_0}^{-1/2} L^2_{I_0}}_{op}\norm{K_{I_0}^{-1/2}(\stb_1)}_{\mbH^{I_0}}}{ \laa\tw_{i}} \\
& \leq  \frac{ \sqrt{\la \tau} \eta_1 \phi \tau \norm{\stb_1}_{\mbK^{I_0}}}{ 2 \laa} \sup\limits_{i \not\in \mcS}(\tw_{i}^{-1})
\leq \frac{\sqrt{\la \tau} \eta_1 \phi \tau I_{0}^{\haf} \max\limits_{i \in \mcS}\norm{\sbi}_{\mbK}}{ 2 \laa} \sup\limits_{i \not\in \mcS}(\norm{\tilde{\beta}_{i} -0}_{\mbH}) =  \frac{\sqrt{\la \tau} \eta_1 \phi \tau I_{0}^{\haf} d_N}{ 2 \laa} r_{N}^{\frac{1}{2}} O_{p}(1).
\end{align*}
where $r_N = \dfrac{I_0 \log{(I)}}{N}$. So if $ \dfrac{I_0 \sqrt{\log(I)}d_N}{\sqrt{N}} \ll \dfrac{\laa}{\sqrt{\la}}$, then $P(B_4) \rightarrow 0$ asymptotically. \\ \\ 

\textbf{step 5:} $P(B_{5}) \rightarrow 0$ \\
The idea is same as Part 3. Recall that
\begin{align} \label{B5:1}
B_{5} =\bigg\{\frac{1}{N} \norm{\bX_{.i}^\top  \bX_1 G_{I_0}(\sbb)}_{\mbH} \geq \dfrac{\tw_{i}}{3}  \hspace{1cm} ; i \not\in \mcS \bigg\},
\end{align}
where  $\sbb=\{ \tilde{w}_{i} \hbi \norm{\hbi}_{\mbH}^{-1} ; i \in \mcS \}$. We aim to show that $ \dfrac{\norm{\bX_{.i}^\top  \bX_1 G_{I_0}(\sbb)}_{\mbH}}{N \tw_{i}} \rightarrow 0$. By using Assumption \ref{a:1.4}, \eqref{p:B3.GK} and \eqref{B3:shat2} and the fact that
\begin{align*}
\sup\limits_{i \not\in \mcS} \tw_{i}^{-1} = \sup\limits_{i \not\in \mcS} \norm{\tilde{\beta}_i-0}_{\mbH}=O_{p}(r_{N}^{\haf}),
\end{align*} 
where $r_N = \dfrac{I_0 \log(I)}{N}$, we can write
\begin{align*}
& \frac{\norm{\bX_{.i}^\top  \bX_1 G_{I_0}(\sbb)}_{\mbH}}{N \tw_{i}} \leq \frac{ \norm{ \hsigg \hsig^{-1} }_{op}  \norm{\hsig}_{op} \norm{ G_{I_0}}_{op} \norm{\sbb}_{\mbH^{I_0}}}{ \tw_{i}} \leq \frac{\phi \tau^2 \sqrt{I_0} \left(1+o_{p}(1) \right)^{\haf}}{b_N \tw_{i}} \\
& \leq \frac{ \phi \tau^2 \sqrt{I_0} \left(1+o_{p}(1) \right)^{\haf}}{b_N} r_{N}^{\haf}O_{p}(1) \leq \frac{ I_0 \sqrt{\log(I)}}{\sqrt{N} b_N} O_{p}(1).
\end{align*}
So if $\dfrac{I_0^2 \log(I)}{N} \ll b_N^2$, then $P(B_5) \rightarrow 0$ asymptotically. \\
 
\textbf{Step 6:} $P(B_{6}) \rightarrow 0$ \\
Recall that 
\begin{align*}
B_{6} =\bigg\{\frac{1}{N} \norm{\bX_{.i}^\top H_N \ep}_{\mbH} \geq \dfrac{\laa\tw_{i}}{3} \hspace{1cm} ; i \not\in \mcS \bigg\},
\end{align*}
where $H_N=\left(\mathds{I}_{N} - \bX_1 \left(\bX_{1}^\top \bX_{1}  \mathds{I}_{I_0} + \la N K_{I_0}^{-1} L^2_{I_0} \right)^{-1} \bX_1^\top \right)$.
We notice that $B_{6} = \bigcup\limits_{i \not\in \mcS} A_{i}$ where $A_{i} =\bigg\{\frac{1}{N}  \norm{\bX_{.i}^\top H_N \epsilon }_{\mbH} \geq \dfrac{\laa\tw_{i}}{3}  \bigg\}$. According to Lemma \ref{l:btilde}
\begin{align}  
\nonumber \\ P(B_6) & \leq \sum\limits_{i \not\in \mcS} P(A_i) = \sum\limits_{i \not\in \mcS} P\left( N^{-1} \norm{\tbi}_{\mbH} \norm{\bX_{.i}^\top H_N \ep }_{\mbH} \geq \frac{\laa}{3}\right) \leq \sum\limits_{i \not\in \mcS} P\left( N^{-1} O_p(r_N^{\haf}) \norm{\bX_{.i}^\top H_N \ep}_{\mbH} \geq \frac{\laa}{3}\right) \nonumber \\ \nonumber \\
& = \sum\limits_{i \not\in \mcS} P\left( O_p(1) \norm{\bX_{.i}^\top H_N \ep}_{\mbH} \geq \frac{N \laa}{3 r_N^{\haf}}\right) = \sum\limits_{i \not\in \mcS} P\left( T \norm{\bX_{.i}^\top H_N \ep}_{\mbH} \geq \frac{N \laa}{3 r_N^{\haf}}\right), \label{B6:prob} 
\end{align}
where we denoted $T=O_p(1)$ which is bounded in probability with $M$ for $\dfrac{\varepsilon}{2 I}$. Using conditional probability on T, we have
\begin{align}
P(B_6) & \leq \sum\limits_{i \not\in \mcS} P\left( T \norm{\bX_{.i}^\top H_N \ep}_{\mbH} \geq \frac{N \laa}{3 r_N^{\haf}} \given[\Big] T > M \right)P(T>M) \nonumber \\ \nonumber  \\
& + \sum\limits_{i \not\in \mcS} P\left( T \norm{\bX_{.i}^\top H_N \ep}_{\mbH} \geq \frac{N \laa}{3 r_N^{\haf}} \given[\Big] T \leq M \right) P(T \leq M) \nonumber \\ \nonumber  \\
& \leq \sum\limits_{i \not\in \mcS} 1 \times P(T>M) + \sum\limits_{i \not\in \mcS} P\left( \norm{\bX_{.i}^\top H_N \ep}_{\mbH} \nonumber \geq \dfrac{N \laa}{3 M r_N^{\haf}}\right) \times 1 \nonumber \\ \nonumber \\
& \leq \frac{\varepsilon}{2} + \sum\limits_{i \not\in \mcS} P\left( \norm{\bX_{.i}^\top H \ep}_{\mbH}^2 \geq ( \frac{N \laa}{3 M r_N^{\haf}} )^2 \right). \nonumber
\end{align}
We just need to show that $ \sum\limits_{i \not\in \mcS} P\left( \norm{\bX_{.i}^\top H_N \ep}_{\mbH}^2 \geq ( \dfrac{N \laa}{3 M r_N^{\haf}} )^2 \right)$ goes to zero.
Since $\ep=(\epsilon_1,\dots,\epsilon_N)$ and $\epsilon_i$s are independent mean zero $C$-subgaussian process in $\mbH$, Lemma \ref{l:multi-sub} implies that 
$\bX_{.i}^\top H_N \ep$ is a $C_h$-subgaussian process where
\begin{align*}
C_{h}= \bX_{.i}^{\top} H_N C_{N} H_N \bX_{.i}.
\end{align*}
According to Lemma \ref{l:subg-Basic} we have
\begin{align} 
P\left(\norm{\bX_{.i}^\top H_N \ep}_{\mbH} \geq ( \norm{C_{h}}_{1} + 2 \norm{C_{h}}_{2} \sqrt{t} + 2 \norm{C_{h}}_{\infty} t )  \right) \leq \exp(-t), \label{l:Chinq2}
\end{align}
and based on Lemma \ref{l:PCP} and the fact that $\bX_{.i}$ is standardized, we have
\begin{align} \label{p:Ch}
\norm{C_h}_m \leq \norm{\bX_{.i}^{\top} H_N}_{op}^2 \norm{C}_m \leq \norm{X_{.i}^{\top}}^2 \norm{H_N}_{op}^2 \norm{C}_m = N \norm{H_N}_{op}^2 \norm{C}_m.
\end{align}
Now we just need to bound $\norm{H_N}_{op}$. Let's denote $P=\left(\bX_{1}^\top \bX_{1} \mathds{I}_{I_0}  + \la N K_{I_0}^{-1} L^2_{I_0} \right)^{-1}$ and then write $H_N =\mathds{I}_N - \bX_1 P \bX_1^\top$. So if we can prove eigenvalues of $\bX_1 P \bX_1^\top$ are in $[0,1]$, it will be obvious the eigenvalues of $H$ will be in $[0,1]$ and consequently $\norm{H_N}_{op} \leq 1$. 
For doing so, we want to use the Lemma \ref{l:idempotent} and prove
\begin{align*}
\norm{\bX_1 P \bX_1^\top \bX_1 P \bX_1^\top x }_{\mbH^N} \leq \norm{\bX_1 P \bX_1^\top x}_{\mbH^N} \hspace{1cm} \forall x \in \mbH^N.
\end{align*}
It is a basic linear algebra exercise to show that $\bX_1 P \bX_1^{\top}$ and $ P \bX_1^{\top} \bX_1$ have the same eigenvalues. Then
\begin{align*}
P \bX_1^\top \bX_1 & = \left(\bX_{1}^\top \bX_{1} \mathds{I}_{I_0}  + \la N K_{I_0}^{-1} L^2_{I_0} \right)^{-1} \bX_1^\top \bX_1 \\ 
& = \left( \mathds{I}_{I_0}  + \la \hsig^{-1} K_{I_0}^{-1} L^2_{I_0} \right)^{-1} = (\bI_{I_0} \otimes \mathds{I} + \la \hsig^{-1} \otimes K^{-1} L^2_{I_0})^{-1},
\end{align*}
where $\mathds{I}$ is an identity operator from $\mbH$ to $\mbH$ and $\bI_{I_0}$ is an identity $I_0$ by $I_0$ matrix. Then we can see 
\begin{align*}
P \bX_1^{\top} \bX_1 (u_i \otimes v_j) = (1+\la \tau_i \theta_j^{-1} \eta_j^{2})^{-1} (u_i \otimes v_j).
\end{align*}
Since $(1+\la \tau_i \theta_j^{-1} \eta_j^{2})^{-1} \leq 1$, the eigenvaluse of $P \bX_1^{\top} \bX_1$
are smaller than one and we can conclude
\begin{align*}
\norm{\bX_1 \left(P \bX_1^\top \bX_1 (P \bX_1^\top x)\right) }_{\mbH^N} \leq \norm{\bX_1 (P \bX_1^\top x)}_{\mbH^N} \hspace{1cm} \forall x \in \mbH^N,
\end{align*}
therefore $\norm{H_N}_{op} \leq 1$. So \eqref{p:Ch} will be simplified to
\begin{align}
\norm{C_h}_m \leq  N \norm{C}_m.
\end{align}
According to \eqref{l:Chinq2}, we can see
\begin{align*}
& P\left(\norm{\bX_{.i}^\top H_N \ep}_{\mbH} \geq N( \norm{C}_{1} + 2 \norm{C}_{2} \sqrt{t} + 2 \norm{C}_{\infty} t )  \right) \leq \exp(-t). 
\end{align*}
So we are looking for a $\hat t$ such that 
\begin{align*}
N (\norm{C}_{1} + 2 \norm{C}_{2} \sqrt{\hat{t}} + 2 \norm{C}_{\infty} \hat{t}) \leq \left(\frac{N \laa}{3 M r_N^{\haf}} \right)^2.
\end{align*}
Same as \eqref{l:CDT}, we have
\begin{align*}
\norm{C}_{1} + 2 \norm{C}_{2} \sqrt{t} + 2 \norm{C}_{\infty} t \leq Dt.
\end{align*}
So we just need to find $\hat{t}$ such that
\begin{align*}
\left(\frac{N \laa}{3 M r_N^{\haf}} \right)^2 \geq N D \hat{t}.
\end{align*}
Let's denote $D_2 = 9 D M^2$, then one can see $\hat t = \dfrac{N \laa^2}{ D_2 r_N}$ implies
\begin{align*}
P\left( \norm{\bX_{.i}^\top H_N \epsilon }_{\mbH} \geq \frac{N \laa}{3M r_N^{\haf}}\right) \leq \exp\left(- \frac{N \laa^2}{ D_2 r_N}\right).
\end{align*}
Based on \eqref{B6:prob}, we can bound
\begin{align}  \label{B6:prob2}
P(B_6) \leq \frac{\varepsilon}{2} + I  \exp\left(- \frac{N \laa^2}{ D_2 r_N}\right) \leq \frac{\varepsilon}{2} + \exp\left(- \frac{N^2 \laa^2}{ D_2 I_0 \log(I)} + \log(I-I_0) \right) \rightarrow 0.
\end{align} 
So if $\laa \gg \dfrac{\sqrt{I_0} \log(I)}{N}$, then we can conclude $P(B_6) \rightarrow 0$ asymptotically.
\\ \\ \\

\textbf{{Proof of part 2:}} \\
We need to show that
\begin{align*}
\sqrt{N}(\hboldb -\oboldb)=o_{P}(1) \quad \quad under \ \ \normh,
\end{align*}
or equivalently, since the support is recovered with probability tending to one,
\begin{align*}
P \left( \sqrt{N} \norm{ \hboldb_{1} - \oboldb_1 }_{\mbH^{I_0}} \geq \varepsilon \right) \rightarrow 0,
\end{align*}
where $\hboldb=(\hboldb_1,\textbf{0})$ and $\oboldb=(\oboldb_1,\textbf{0})$. So 
\begin{align*}
\hboldb_{1} & = \hsig G_{I_0} (\stb_{1}) + G_{I_0} \left(  N^{-1} \bX_{1}^\top \ep - \laa \sbb \right), \\
\oboldb_1 & = \left( \hsig  \mathds{I}_{I_0}  + \la K_{I_0}^{-1} L^2_{I_0} \right)^{-1} \left(\frac{1}{N} \bX_{1}^\top \bY \right) = \hsig G_{I_0} (\stb_{1}) + G_{I_0} \left(  N^{-1} \bX_{1}^\top \ep \right),
\end{align*}
where $G_{I_0} = \left(\hsig \mathds{I}_{I_0}+\la K_{I_0}^{-1} L_{I_0}^2 \right)^{-1}$ and the oracle estimate is obtained by taking the subgradient of target function \eqref{AFSSEN} with $\laa = 0$ given the true support.
So the norm of the difference is given by
\begin{align} \label{p:ineq-GKK}
\sqrt{N} \norm{ \hboldb_{1} - \oboldb_1 }_{\mbH^{I_0}} =\sqrt{N} \laa \norm{ G_{I_0} \sbb} _{\mbH^{I_0}} \leq \sqrt{N} \laa \norm{ G_{I_0}}_{op} \norm{\sbb}_{\mbH^{I_0}}.
\end{align}
According to \eqref{t:GI0} and \eqref{B3:shat2} we will have
\begin{align*}
&  \norm{G_{I_0}}_{op} \leq \tau, \\
& \norm{\sbb}_{\mbH^{I_0}} \leq \frac{I_{0}^{1/2}}{b_{N}} \left(1+o_{p}(1) \right),
\end{align*}
then we can conclude
\begin{align*}
\sqrt{N} \norm{ \hboldb_{1} - \oboldb_1 }_{\mbH^{I_0}} \leq  \sqrt{N} \laa \tau \frac{I_0^{1/2}}{b_N} \left(1+o_{p}(1) \right).
\end{align*}
So if $\laa \ll \dfrac{b_N}{\sqrt{N} \sqrt{I_0}}$, the probability asymptotically goes to zero.
\\ \\ \\
\textbf{{Proof of part 3:}} \\
Here we want to show that
\begin{align*}
\sqrt{N}(\hboldb -\oboldb)=o_{P}(1) \quad \quad under \ \ \normk,
\end{align*}
or equivalently, since the correct support is recovered with probability tending to one,
\begin{align*}
P \left( \sqrt{N} \norm{ \hboldb_{1} - \oboldb_{1} }_{\mbK^{I_0}} \geq \varepsilon \right) \rightarrow 0.
\end{align*}
Similar to \eqref{p:ineq-GKK}, we can see
\begin{align*}
& \sqrt{N} \norm{\hboldb_{1} - \oboldb_{1}}_{\mbK^{I_0}} = \sqrt{N} \laa \norm{ G_{I_0} \sbb}_{\mbK^{I_0}}  =  \sqrt{N} \laa \norm{ K_{I_0}^{-1/2} G_{I_0}\sbb} _{\mbH^{I_0}} \\ \\
& =   \sqrt{N} \laa \norm{ K_{I_0}^{-1/2} \left(\hsig \mathds{I}_{I_0}+\la K_{I_0}^{-1} L_{I_0}^2 \right)^{-1}  \sbb} _{\mbH^{I_0}} \leq  \sqrt{N} \laa \norm{ \left( \hsig K_{I_0}^{1/2} + \la K^{-1/2}_{I_0} L_{I_0}^2 \right)^{-1}}_{op} \norm{\sbb}_{\mbH^{I_0}}.
\end{align*}
Since
\begin{align*}
\left(\hsig K_{I_0}^{1/2} + \la K^{-1/2}_{I_0} L_{I_0}^2 \right)^{-1} = \left( \hsig \otimes K^{1/2} + \la \bI_{I_0} \otimes K^{-1/2} L^2 \right)^{-1},
\end{align*}
where $\bI_{I_0}$ is an identity $I_0$ by $I_0$ matrix. Then we can see
\begin{align*}
\left( \hsig \otimes K^{1/2} + \la \bI_{I_0} \otimes K^{-1/2} L^2 \right)^{-1}(u_i \otimes v_j) = (\tau_i \theta_j^{1/2} + \la \theta_j^{-1/2} \eta_j^2)^{-1} (u_i \otimes v_j).
\end{align*}
Then we can conclude
\begin{align*}
\norm{\left( \hsig \otimes K^{1/2} + \la \bI_{I_0} \otimes K^{-1/2} L^2 \right)^{-1}}_{op} = \max \dfrac{1}{(\tau_i \theta_j^{1/2} + \la \theta_j^{-1/2} \eta_j^2)}.
\end{align*}
If there exists a constant $M>0$ such that $\eta_j^2 \geq M \sqrt{\theta_j}$, we have
\begin{align*}
\norm{\left( \hsig \otimes K^{1/2} + \la \bI_{I_0} \otimes K^{-1/2} L^2 \right)^{-1}}_{op} \leq \dfrac{1}{\la M},
\end{align*}
and then based on \eqref{B3:shat2}, we can see
\begin{align*}
\sqrt{N} \norm{\hboldb_{1} - \oboldb_{1}}_{\mbK^{I_0}} \leq \dfrac{\sqrt{N} \laa}{M \la} \norm{\sbb}_{\mbH^{I_0}} \leq \dfrac{\sqrt{N} \laa}{M \la} \frac{\sqrt{I_{0}}}{b_{N}} \left(1+o_{p}(1) \right)^{\haf}.
\end{align*}
So if $\dfrac{\laa}{\la} \ll \dfrac{b_N}{\sqrt{N} \sqrt{I_0}}$, the proof will be completed.

\end{document}